\documentclass[journal,final]{IEEEtran}
\pdfoutput=1

\usepackage{latexsym}
\usepackage{caption}
\usepackage{graphicx}
\usepackage{epstopdf}
\usepackage{mathptmx}
\usepackage[shortlabels]{enumitem}
\usepackage{mathtools}
\usepackage{amsmath}
\usepackage{amsfonts}
\usepackage{amssymb}
\usepackage{amsbsy}
\usepackage{mathrsfs}  
\usepackage{amsthm}
\usepackage{bbm}
\usepackage{dsfont}
\usepackage{physics}
\usepackage[ruled]{algorithm2e}
\usepackage{hyperref}
    \hypersetup{colorlinks=true,urlcolor=blue,citecolor=black,anchorcolor=black,linkcolor=black,bookmarks=false}
\usepackage[capitalise]{cleveref}
\crefname{appsec}{Appendix}{Appendices}
\usepackage{cuted}
\usepackage{subcaption}
\usepackage{multirow}

\ifCLASSINFOpdf
\else
\fi
\theoremstyle{plain}

\newtheorem{theorem}{Theorem}

\newtheorem{corollary}[theorem]{Corollary}

\newtheorem{remark}{Remark}

\newtheorem{lemma}{Lemma}

\def\bx{ {\bf x} }
\def\by{ {\bf y} }

\def\bX{ {\bf X} }

\def\ba{ {\bf a} }

\begin{document}
%
\title{An Optimal Computing Budget Allocation Tree Policy for Monte Carlo Tree Search}
%
%
%

\author{Yunchuan Li,~\IEEEmembership{}
        Michael C. Fu~\IEEEmembership{}
        and Jie Xu~\IEEEmembership{}
        \thanks{This work was supported in part by the National Science Foundation under Grant CMMI-1434419 and DMS-1923145, the Air Force Office of Scientific Research under Grant FA9550-19-1-0383 and by the Defense Advanced Research Projects Agency (DARPA) under Grant  N660011824024. The views, opinions, and/or findings expressed are those of the authors and should not be interpreted as representing the official views or policies of the Department of Defense or the U.S. Government. A preliminary version of this work \cite{li2019monte} was published in the proceedings of the 2019 IEEE Conference on Decision and Control. }
\thanks{Y. Li is with the Department
of Electrical and Computer Engineering and the Institute of Systems Research, University of Maryland, College Park USA, e-mail: yli93@terpmail.umd.edu.} 
\thanks{M. C. Fu is with the R. H. Smith School of Business and the Institute of Systems Research, University of Maryland, College Park USA, e-mail: mfu@umd.edu. }
\thanks{Jie Xu is with the Department of Systems Engineering and Operations Research, George Mason University, Fairfax VA 22030, USA, email: jxu13@gmu.edu. }
}

%
%

\markboth{IEEE Transactions On Automatic Control, Vol. XX, No. Y, Month
	Year}%
{Li, Fu and Xu: article title}
%



\maketitle

\begin{abstract}
	We analyze a tree search problem with an underlying Markov decision process, in which the goal is to identify the best action at the root that achieves the highest cumulative reward. We present a new tree policy that optimally allocates a limited computing budget to maximize a lower bound on the probability of correctly selecting the best action at each node. Compared to widely used Upper Confidence Bound (UCB) tree policies, the new tree policy presents a more balanced approach to manage the exploration and exploitation trade-off when the sampling budget is limited. Furthermore, UCB assumes that the support of reward distribution is known, whereas our algorithm relaxes this assumption. Numerical experiments demonstrate the efficiency of our algorithm in selecting the best action at the root.
\end{abstract}


%
\IEEEpeerreviewmaketitle

\section{Introduction}
We consider a reinforcement learning problem where an agent interacts with an underlying environment. A Markov Decision Process (MDP) with finite horizon is used to model the environment. In each move, the agent will take an action, receive a reward and land in a new state. The reward is usually random, and its distribution depends on both the state of the agent and the action taken. The distribution of the next state is also determined by the agent's current state and action. Our goal is to determine the optimal sequence of actions that leads to the highest expected reward. The optimality of the decision policy will be evaluated by the probability of correctly selecting the best action in the first stage of the underlying MDP. \par 

If the distributions and the dynamics of the environment are known, the optimal set of actions can be computed through dynamic programming \cite{bertsekas1995dynamic}. Under more general settings where the agent does not have perfect information regarding the environment, \cite{chang2005adaptive} proposed an adaptive algorithm based on a Multi-Armed Bandit (MAB) model and Upper Confidence Bound (UCB) \cite{auer2002finite}. 
\cite{kocsis2006bandit} and \cite{coulom2007efficient} applied UCB to tree search, and \cite{coulom2007efficient} invented the term Monte Carlo Tree Search (MCTS) and used it in a Go-playing program for the first time. Since then, MCTS has been developed extensively and applied to various games such as Othello \cite{hingston2007experiments} and Go \cite{silver2016mastering}.  To deal with different types of problems, several variations of MCTS have been introduced, e.g., Flat UCB (and its extension Bandit Algorithm for Smooth Trees) \cite{coquelin2007bandit} and Single-Player MCTS (for single-player games) \cite{schadd2011selective}.\par 

However, most  bandit-based MCTS algorithms are designed to minimize regret (or maximize the cumulative reward of the agent), whereas in many situations, the goal of the agent may be to efficiently determine the optimal set of actions within a limited sampling budget. To the best of our knowledge, there is limited effort in the literature that aims at addressing the latter problem. \cite{teraoka2014efficient} first incorporated Best Arm Identification (BAI) into MCTS for a MIN-MAX game tree, and provided upper bounds of play-outs under different settings. 
\cite{kaufmann2017monte} had an objective similar to \cite{teraoka2014efficient}, but with a tighter bound. Their tree selection policy selects the node with largest confidence interval, which can be seen as choosing the node with the highest variance. In some sense, this is a pure exploration policy and would not efficiently use the limited sampling budget. In our work, we are motivated to establish a tree policy that intelligently balances exploration and exploitation (analogous to the objective of UCB). The algorithms developed in \cite{teraoka2014efficient} and \cite{kaufmann2017monte} are only for MIN-MAX game trees, whereas our new tree policy can be applied to more general types of tree search problems. The MCTS algorithm in \cite{grill2016blazing} is more general than \cite{teraoka2014efficient} and \cite{kaufmann2017monte}, but its goal is to estimate the maximum expected cumulative reward at the root node, whereas we focus on identifying the optimal action.  \par

Algorithms that focus on minimizing regret tend to discourage exploration. This tendency can be seen in two ways. Suppose at some point an action was performed and received a small reward. To minimize regret, the algorithm would be discouraged from taking this action again. However, the small reward could be due to the randomness in the reward distribution. Mathematically, \cite{lai1985asymptotically} showed that for MAB algorithms, the number of times the optimal action is taken is exponentially more than sub-optimal ones, which makes sense when the objective is to maximize the cumulative reward, since the exploration of other actions is highly discouraged. This leads to our second motivation: is there a tree policy that explores sub-optimal actions more to ensure the optimal action is found?\par

Apart from the lack of exploration as a result of the underlying MAB model's objective to minimize regret or maximize cumulative reward, most MCTS algorithms assume that the support of the reward distribution is bounded and known (typically assumed to be \([0,1]\)). With the support of reward distribution being known, the parameter in the upper confidence term in UCB is tuned or the reward is normalized. However, a general tree search problem may likely have an unknown and practically unbounded range of rewards. In such case, assuming a range can lead to very poor performance. Therefore, the third motivation of our research is to relax the known reward support assumption.\par 

To tackle the challenge in balancing exploration and exploitation with a limited sampling budget for a tree policy, we model the tree selection problem at each stage as a statistical {\em Ranking \& Selection} (R\&S) problem and propose a new tree policy for MCTS based on an adaptive algorithm from the R\&S community. Similar to the MAB problem, R\&S assumes that we are given a set of bandit machines (often referred to as alternatives in the R\&S literature) with unknown reward distributions, and the goal is to select the machine with the highest mean reward. Specifically, we will develop an MCTS tree policy based on the Optimal Computing Budget Allocation (OCBA) framework \cite{chen2000simulation}. OCBA was first proposed in \cite{chen1995effective}, and aims at maximizing the probability of correctly selecting the action with highest mean reward using limited sampling budget. More recent developments of OCBA include addressing multiple objectives \cite{lee2004optimal} and subset selection \cite{chen2008efficient,zhang2016simulation}.\par  

The objective of the proposed OCBA tree policy is to maximize the Approximate Probability of Correct Selection (APCS), which is a lower bound on the probability of correctly selecting the optimal action at each node.
Intuitively, the objective function of the new OCBA tree selection policy would lead to an optimal balance between exploration and exploitation with a limited sampling budget, and thus help address the drawbacks of existing work that either pursues pure exploration \cite{teraoka2014efficient,kaufmann2017monte} or exponentially discourages exploration \cite{lai1985asymptotically}. Our new OCBA tree policy also removes the known and bounded support assumption for the reward distribution, because the new OCBA policy determines the sampling allocation based on the posterior distribution of each action, which is updated adaptively according to samples.\par

To summarize, contributions of this paper include the following:
\begin{enumerate}
	\item We propose a new tree policy for MCTS with an objective to maximize APCS with a limited sampling budget. The new tree policy optimally balances exploration and exploitation to efficiently select the optimal action. The new OCBA tree selection policy also relaxes the assumption of known bounded support on the reward distribution.
	\item We present a sequential algorithm to implement the new OCBA tree policy that maximizes the APCS at each sampling stage and prove that our algorithm converges to the optimal action.
	\item We provide theoretical analyses, such as convergence guarantee, proof of optimality of the proposed algorithm, and the exploration-exploitation trade-off of the proposed algorithm, which works differently than bandit-based algorithms, and is more suitable for identifying the best action.
	\item We demonstrate the efficiency of our algorithm through numerical experiments.
\end{enumerate}\par

\begin{remark}
	In much of the computer science/artificial intelligence literature, an algorithm that focuses on determining the optimal set of actions under a limited budget is defined as a pure exploration algorithm (see, e.g., \cite{shleyfman2015interruptible,chen2014combinatorial,bubeck2009pure}), whereas we view such algorithms as retaining a balance between exploration and exploitation, as the analysis in \Cref{sec:algorithm_description} shows. In statistical R\&S,
	pure exploration algorithms generally implies sampling based primarily on the variance of each action, which often leads to sampling suboptimal actions more. 
	It will be clearer in the \Cref{sec:simulation} where we show that OCBA-MCTS actually samples less those highly suboptimal actions and ``exploits" those potential actions more.
\end{remark}
The rest of the paper is organized as follows. We present the problem formulation in \Cref{sec:background}, and review the proposed OCBA-MCTS algorithm in Section \ref{sec:algorithm_description}. Theoretical analyses, including convergence theorems and exploration-exploitation analysis, are carried out in \Cref{sec:analysis}. Proofs are given in the Appendix. Numerical examples are presented in Section \ref{sec:simulation} to evaluate the performance of our algorithm. \Cref{sec:conclusion} concludes the paper and points to future research directions. \par

A preliminary version of this work was presented in \cite{li2019monte}, where a simpler tree policy (not employing the node representation adopted in the current work) was used. In addition to improving the efficiency of the MCTS algorithm, here we prove that our proposed algorithm converges asymptotically to the optimal action, and provide an exploration-exploitation trade-off analysis, both analytically and through a more comprehensive set of numerical experiments. \par  

\section{Problem formulation}\label{sec:background}
Consider a finite horizon MDP \(M = (X,A,P,R)\) with horizon length $H$, finite state space \(X\), finite action space \(A\) with \(|A|>1\), bounded reward function \(R = \{R_t, t=0,1,\dots H \}\) such that \(R_t\) maps a state-action pair to a random variable (r.v.),
and transition function \(P = \{P_t, t=0,1,\dots H \}\) such that \(P_t\) maps a state-action pair to a probability distribution over \(X\). We assume that \(P_t\) is unknown and/or \(|X|\) and \(|A|\) are very large, and hence it is not feasible to solve the problem by dynamic programming. Further define $X_a$ and $A_x$ as the available child states when taking action $a$ and available actions at state $x$, respectively. Denote by \(P_t(x,a)(y)\) the probability of transitioning to state \(y\in X_a\) from state \(x \in X\) when taking action \(a\in A_x \) in stage $t$, and $R_t(x, a)$ the reward in stage $t$ by taking action $a$ in state $x$.
Let \(\Pi\) be the set of all possible nonstationary Markovian policies \(\pi = \{\pi_i|\pi_i:X\rightarrow A, i\ge0\}\).

Bandit-based algorithms for MDPs seek to minimize the expected cumulative regret,
whereas our objective is to identify the best action that leads to maximum total expected reward given by $\mathbb{E}\big[\sum_{t = 0}^{H-1} R_t(x_t,\pi_t(x_t)) \big]$ for given $x_0 \in X$. We first define the optimal reward-to-go value function for state \(x\) in stage \(i\) by 
\begin{align}\label{eq:optimal-to-go}
V_i^*(x) =
\max_{\pi\in\Pi}\mathbb{E}\big[\sum_{t = i}^{H-1} R_t(x_t,\pi_t(x_t)) \big| x_i = x \big], \nonumber \\
~ i = 0,1,\dots,H-1
\end{align}
with \(V_H^*(x) = 0\) for all \(x\in X\). Also define 
\begin{align*}
Q_i(x,a) = \mathbb{E}[R(x,a)] + \sum_{y\in X_a}P_t(x,a)(y)V_{i+1}^*(y),
\end{align*}
with \(Q_H(x,a) = 0\). It is well known \cite{bertsekas1995dynamic} that eq. (\ref{eq:optimal-to-go}) can be written via the standard Bellman optimality equation:
\begin{align*}\label{eq:optimal-to-go_Bellman}
V_i^*(x) &=
\max_{a \in A_x} (\mathbb{E}[R_i(x,a)] + \mathbb{E}_{P_t(x,a)}V_{i+1}^*(Y)), ~  \\
&= 	\max_{a \in A_x} (\mathbb{E}[R_i(x,a)] + \sum_{y\in X_a}P_t(x,a)(y)V_{i+1}^*(y) )\\
&= \max_{a \in A_x}(Q_i(x,a)), ~ i = 0,1,\dots,H-1,
\end{align*}
where \(Y\sim P_i(x,a)(\cdot)\) represents the random next state.\par 

Since we are considering a tree search problem, some additional notation and definitions beyond MDP settings are needed. Define a state node by a tuple that contains the state and the stage number:
\begin{align*}
    &\bx = (x, i) \in \bX \\
    &\forall x \in X, ~ 0 \le i \le H,
\end{align*}
where $\bX$ is the set of state nodes.
Similarly, we define a state-action node by a tuple of state, stage number and action (i.e., a state node followed by an action):
\begin{align*}
    & \ba = (\bx, a) = (x, i, a),\\
    &\forall x \in X, ~ 0 \le i \le H, ~ a \in A_{x},
\end{align*}

Now, we can rewrite the immediate reward function, value function for state, state-action pair with state node and state-action node and state transition distribution, respectively, by 
\begin{align*}
    R(\ba) &= R(\bx, a)  :=  R_i(x, a)\\
    V^*(\bx) &:= V^*_i(x), \\
    Q(\ba) &= Q(\bx, a) := Q_i(x, a)\\ 
    P(\ba) &= P(\bx, a) := P_i(x, a).
\end{align*}
Similarly, $V^*(\bx)$ and $Q(\bx, a)$ are assumed to be zero for all terminal state nodes $\bx$.
To make our presentation clearer, we adopt the following definitions based on nodes: define $N(\bx)$ and $N(\bx, a)$ the number of visits to node $\bx$ and $(\bx, a)$, respectively, $\bX_\ba$ the set of child state nodes given parent nodes, and $A_{\bx}$ the set of available child actions at node $\bx$, respectively. 

Traditionally, MCTS algorithms aim at estimating \(V^*(\bx)\) and model the selection process in each stage as an MAB problem, i.e., view  \(Q(\bx, a)\) as a set of bandit machines where $(\bx, a)$ are child state-action nodes of $\bx$ (\cite{chang2005adaptive,kocsis2006bandit}), and minimize the {\it regret}, namely,
\begin{align*}
\min_{a_1, \dots, a_N \in A_{\bx}} &\{ N \max_{a \in A_{\bx}}(Q(\bx, a))- \sum_{k = 1}^{N} Q(\bx, a_k)  \}\\
&=\{ N V^*(\bx)- \sum_{k = 1}^{N} Q(\bx, a_k)  \}
\end{align*}
for $\bx$ in stage $1,2,\dots H$,
where $N$ and $a_k$ are the number of rollouts/simulations (also known as total sampling budget in much of Ranking \& Selection literature) and the $k$-th action sampled at state node $\bx$ by the tree policy, respectively. The meaning of rollout will be clearer in \Cref{sec:algorithm_description}.
In this paper, our goal is to identify the optimal action that achieves the highest cumulative reward at the root with initial state \(x\), that is, find
\begin{align*}
a_{\bx_0}^* = \arg \max_{a \in A_{\bx_0}} Q(\bx_0,a),
\end{align*}
where the root state node $\bx_0 = (x, 0)$.
Let \(\hat{Q}(\bx, a) =R(\bx, a) + V^*(\by) \) be the random cumulative reward by taking action \(a\) at state node \(\bx\), where $\by$ is the random state node reached.
Clearly, \(\hat{Q}(\bx, a)\) is a random variable. We assume \(\hat{Q}(\bx, a)\) is normally distributed with known variance, and its mean \(\mu(\bx, a)\) has a conjugate normal prior with a mean equals \(Q(\bx, a)\). Hence we have
\begin{align*}
Q(\bx, a) = \mathbb{E}[\mathbb{E}[\hat{Q}(\bx, a)|\mu(\bx, a)]].
\end{align*}

\begin{remark}
	For our derivations, we assume the variance of the sampling distribution of $\hat{Q}(x,a)$ is known; however, in practice, the prior variance may be unknown, in which case estimates such as the sample variance are used \cite{chen2010stochastic}.
\end{remark}
Consider the non-informative case, i.e., the prior mean $Q(\bx, a)$ is unknown, it can be shown that \cite{degroot2005optimal} the posterior of \(\mu(\bx, a)\) given observations (i.e., samples) is also normal. 
For convenience, define the \(t\)-th sample by \(\hat{Q}^t(\bx, a)\).  Then the conditional distribution of \(\mu(\bx, a)\) given the set of samples \((\hat{Q}^1(\bx, a), \hat{Q}^2(\bx, a),\dots, \hat{Q}^{N(\bx, a)}(\bx, a))\) is 
\begin{equation}\label{eq:posterior}
\tilde{Q}(\bx, a) \sim N(\bar{Q}(\bx, a) , \frac{\sigma^2(\bx, a)}{N(\bx, a)}),
\end{equation}
where
\begin{align*}
\bar{Q}(\bx, a) &= \frac{1}{N(\bx, a)}\sum_{t = 1}^{N(\bx, a)}\hat{Q}^t(\bx, a),\\
\tilde{Q}(\bx, a) &= \mu(\bx, a)|(\hat{Q}^1(\bx, a), \hat{Q}^2(\bx, a),\dots, \hat{Q}^{N(x,a)}(\bx, a)),
\end{align*}
and \(\sigma^2(\bx, a)\) is the variance of \(\hat{Q}(\bx, a)\) and can be approximated by the sample variance:
\begin{align*}
\hat{\sigma}^2(\bx, a) &=  \frac{1}{N(\bx, a)} \sum_{t = 1}^{N(\bx, a)}\big(\hat{Q}^t(\bx, a) - \bar{Q}(\bx, a)\big)^2.
\end{align*}

\begin{remark}
	If the samples of \(Q(\bx, a)\) are not normally distributed, the normal assumption can be justified by batch sampling and the central limit theorem.
\end{remark}
Under these settings, our objective is to maximize the Probability of Correct Selection (PCS) defined by
\begin{align}
PCS &= P\bigg[ \bigcap_{a\in A, a\ne \hat{a}_{\bx}^*} (\tilde{Q}(\bx,\hat{a}_{\bx}^*) \ge \tilde{Q}(\bx, a) ) \bigg]
\end{align}
for a state node $\bx$, where \(\hat{a}_{\bx}^*\) is the action that achieves the highest mean sample \(Q\)-value at such node, i.e., \(\hat{a}_{\bx}^* =\arg \max_{a\in A_{\bx}} \bar{Q}(\bx, a)\). \par

PCS is hard to compute because of the intersections in the (joint) probability. We seek to simplify the joint probability by changing the intersections to sums using the Bonferroni inequality to make the problem tractable.
By the Bonferroni inequality, PCS is lower bounded by the Approximate Probability of Correct Selection (APCS), that is,
\begin{align}\label{eq:APCS}
PCS \ge& 1-\sum_{a\in A_{\bx}, a\ne \hat{a}_{\bx}^*} P\bigg[ \tilde{Q}(\bx,\hat{a}_{\bx}^*) \le \tilde{Q}(\bx, a) \bigg]\\
=\vcentcolon& APCS. \nonumber
\end{align} \par 
The objective of our new tree policy is to maximize APCS as given in \Cref{eq:APCS}. Compared to MAB's objective of minimizing the expected cumulative regret, this objective function will result in an allocation of sampling budget to alternative actions in a way that optimally balances exploration and exploitation. This objective function is motivated by the OCBA algorithm \cite{chen2000simulation} in the R\&S literature. We will present and analyze our OCBA tree policy in the following sections. 

\section{Algorithm description}\label{sec:algorithm_description}

In this section, we first briefly describe the main four phases, i.e., {\it selection}, {\it expansion}, {\it simulation} and {\it backpropagation}, in an MCTS algorithm. Then, we propose a novel tree policy in the selection stage that aims at finding the optimal action at each state node. 

\subsection{Canonical MCTS algorithm}\label{sec:canonical_MCTS}
Here we briefly summarize the four phases in a typical MCTS algorithm. We refer readers to \cite{browne2012survey} for a complete illustration of these phases. 
Algorithm \ref{alg:MCTS} represents a canonical MCTS, with detailed descriptions of the main phases below.

\subsubsection{Selection}
In this phase, the algorithm will navigate down the tree from the root state node to an expandable node, i.e., a node with unvisited child nodes. We assume that expansion is automatically followed when a state-action is encountered. Therefore, when determining the path down, there are three possible situations:
\begin{enumerate}[(i)]
    \item If a state-action node is encountered (denoted by $(\bx, a)$), we will land into a new state node $\by$ which is obtained by calling the expansion function. Then, we continue with the selection algorithm.
    
    \item If an expandable state node (which could be a leaf node) is encountered, we call the expansion function to add a new child state-action node and a state node (by automatically expanding the state-action node) to the path. Then, we stop the selection phase and return the path from the root to this state node. Finally, we proceed with the simulation and backpropagation phase
    \item If an unexpandable state node is encountered (denoted by $\bx$), we employ a {\it tree policy} to determine which child action to sample. Then we enter the new state-action node $(\bx, a)$ and continue the selection algorithm with this state-action node. The tree policies can be briefly categorized into two types: deterministic, such as UCB1 and several of its variants (e.g., UCB-tuned, UCB-E), and stochastic, such as $\epsilon$-greedy and EXP3; see \cite{browne2012survey} for a review. 
    
\end{enumerate}

\subsubsection{Expansion}
In this phase, a random child state or state-action node of the given node is added. 
If the incoming node is a state node $\bx$, the next node is selected randomly (usually uniform) from those unvisited child state-action nodes.
If the incoming node is a state-action node $(\bx, a)$, the subsequent state node is found by simply sampling from distribution $P(\bx, a)(\cdot)$.

\subsubsection{Simulation}
In some literature, this phase is also known as ``rollout". The simulation phase starts with a state node. The purpose of this step is to simulate a path from this node to a terminal node and produce a sample of cumulative reward by taking this path (which is a sample of the value for this node). The simulated path is taken by a {\it default policy}, which is usually sample the feasible child sate-action nodes uniformly. With this node's value sample, we may proceed to the backpropagation phase. 

\subsubsection{Backpropagation}
This phase simply takes the simulated node value and update the values of the nodes in the path (obtained in selection step) backward.

In the next section, we will propose our tree policy based on OCBA and illustrate the detailed implementations of the four phases.

\subsection{OCBA selection algorithm}
We now present an efficient tree policy to estimate the optimal actions in every state node by estimating \(V^*(\bx)\) and \(Q(\bx, a)\) for all possible \(a \in A_{\bx}\) at the state node. 
Denote the estimates of \({V}^*(\bx)\) at node \(\bx\) by \(\hat{V}^*(\bx)\), which is initialized to 0 for all state nodes. Our algorithm estimates \(Q(\bx, a)\) for each action \(a\) by its sample mean, and selects the action that maximizes the sample mean as \(\hat{a}_{\bx}^*\). During the process, the estimate of \(Q(\bx, a)\) is given by \Cref{eq:posterior} and the proposed new OCBA tree policy is applied. Our algorithm follows the algorithmic framework described in \Cref{sec:canonical_MCTS}, with the tree policy changed to OCBA and other mild modifications. 

The structure of the proposed OCBA-MCTS algorithm is shown in Algorithms \ref{alg:MCTS} to \ref{alg:backpropagate}. There are two major characteristics: the first is to use the proposed OCBA algorithm for the tree policy. The second is to require each state-action node to be expanded $n_0 > 1$ times, because we need a sample variance for each state-action node, which will become clearer after the tree policy illustration. The process is run for a prespecified \(N\) times (which will be later referred to as number of rollouts or sampling budget) from the root state node $\bx_0$, after which a partially expanded tree is obtained and the optimal action \(\hat{a}_{\bx_0}^*\) can be derived.\par 

When steering down the tree and a state node $\bx$ is visited, the selection phase, which is illustrated in Algorithm \ref{alg:selection}, will first determine if there is a child state-action node that was visited for less than $n_0$ times at the given state node. If there is, then the state-action node will be sampled and added to the path. In other words, we try to expand each state node when it is visited, and require each node to be expanded $n_0$ times. If all the state-action nodes are well-expanded, Algorithm \ref{alg:selection} will call Algorithm \ref{alg:OCBAselection} (OCBASelection), which calculates the allocation of samples to child state-action nodes of the current state node for a total sampling budget \(\sum_{a \in A} N(\bx, a) +1 \). To determine the number of samples allocated to each state-action node, denoted by \((\tilde{N}(\bx,a_1),\tilde{N}(\bx, a_2),\dots,\tilde{N}(\bx, a_{|A_{\bx}|}))\) (where $a_i \in A_{\bx}, i = 1,\dots,|A_{\bx}|$), the OCBA tree policy first identifies the child state-action node with the largest sample mean (sample optimal) and finds the difference between the sample means of the sample optimum and all other nodes:
\begin{align*}
\hat{a}_{\bx}^* &:= \arg\max_a \bar{Q}(\bx, a)\\
\delta_{\bx}(\hat{a}_{\bx}^*,a) &:=\bar{Q}(\bx, \hat{a}_{\bx}^*)-\bar{Q}(\bx, a), ~\forall a \ne \hat{a}_{\bx}^*.
\end{align*}
The set of allocations \((\tilde{N}(\bx, a_1),\tilde{N}(\bx, a_2),\dots,\tilde{N}(\bx, a_{|A|}))\) {that maximizes APCS} can be obtained by solving the following set of equations:

\begin{align}
\frac{\tilde{N}(\bx,a_{n+1})}{\tilde{N}(\bx, a_n)} =& 
\Bigg(\frac{\sigma(\bx, a_{n+1})/\delta_{\bx}(\hat{a}_{\bx}^*,a_{n+1})  } {\sigma(\bx,  a_n)/\delta_{\bx}(\hat{a}_{\bx}^*,a_n) }\Bigg)^2, \nonumber \\
~ &\forall  a_n, a_{n+1} \ne \hat{a}_{\bx}^*, ~ a_n, a_{n+1} \in A_{\bx}, \label{eq:budget_allocation1}\\
\tilde{N}(\bx, \hat{a}_{\bx}^*) =& \sigma(\bx, \hat{a}^*_{\bx})\sqrt{\sum_{a \in A, a\ne \hat{a}_{\bx}^*}  \frac{(\tilde{N}(\bx, a))^2}{\sigma^2(\bx, a) } }, \label{eq:budget_allocation2}\\
\sum_{a \in A}\tilde{N}(\bx, a) =& \sum_{a \in A}N(\bx, a) +1 \label{eq:budget_allocation3}.
\end{align}	
The derivations of \Crefrange{eq:budget_allocation1}{eq:budget_allocation3} are illustrated in the appendix.\par 

After the new budget allocation is computed, the algorithm will select the ``most starving'' action to sample \cite{chen2010stochastic}, i.e., sample 
\begin{align}\label{eq:select_policy}
\hat{a} = \arg\max_{a \in A_{\bx}}(\tilde{N}(\bx, a) - N(\bx, a) ).
\end{align} 

We highlight some major modifications to the canonical MCTS in the proposed algorithm. First, in the selection phase, we will try to expand all ``expandable" nodes visited when obtaining a path to leaf. Since the variances of the values of a state node's child nodes are required in the proposed tree policy, we define a state node as expandable if it has child nodes that are visited less than $n_0 > 1$ times. State-action nodes are always expandable.\par 

At the expansion phase as shown in Algorithm \ref{alg:expand}, a state-action node is expanded by simply sampling the transition distribution \(P(\bx, a)(\cdot)\), and the resulting state node is subsequently added to the path. The reward by taking the action in the state node is also recorded and will be used in the backpropagation stage.\par 

In the simulation and backpropagation phases illustrated in Algorithm \ref{alg:simulate} and \ref{alg:backpropagate}, a leaf-to-terminal path is simulated, and its reward is used to update the value for the leaf node. If we denote the leaf node and the reward from the simulated path by $\bx_l$ and $r$, respectively, the leaf node value estimate is updated by
\begin{align}\label{eq:bp_leaf}
    \hat{V}^*(\bx_l) \leftarrow \frac{N(\bx_l) - 1}{N(\bx_l)} \hat{V}^*(\bx_l) + \frac{1}{N(\bx_l)} r.
\end{align}
After updating the leaf state node, we update the nodes in the path collected in selection stage in reversed order. Suppose we have a path 
\begin{align*}
    (\bx_0, (\bx_0, a_0), \dots, \bx_i, (\bx_i, a_i), \bx_{i+1}, \dots, \bx_l)
\end{align*}
and the node values of $\bx_{i+1}, \dots, \bx_l$ have been updated, the preceding nodes $\bx_i$ and $(\bx_i, a_i)$ are updated through
\begin{align}
    \hat{Q}^{N(\bx, a)}(\bx_i, a_i) &= R(\bx_i, a) + \hat{V}^*(\bx_{i+1}), \label{eq:bp1} \\
    \bar{Q}(\bx_i, a_i) &\leftarrow \frac{N(\bx_i, a_i) - 1}{N(\bx_i, a_i)} \bar{Q}(\bx_i, a_i) + \frac{1}{N(\bx_i, a_i)} Q^{N(\bx, a)}(\bx_i, a_i), \label{eq:bp2}\\
    \bar{V}(\bx_i) &\leftarrow \frac{N(\bx_i) - 1}{N(\bx_i)} \bar{V}^*(\bx_i) + \frac{1}{N(\bx_i)} \bar{Q}(\bx_i, a_i), \label{eq:bp3}\\
    \hat{V}(\bx_i) & \leftarrow (1-\alpha_{N(\bx_i)}) \bar{V}(\bx_i) + \alpha_{N(\bx_i)} \max_{a\in A_{\bx_i}} \bar{Q}(\bx_i, a), \label{eq:bp4}
\end{align}
where $\bar{V}(\cdot)$ is an intermediate variable that records the average value of the node through the root-to-leaf path, and $\alpha_{N(\bx_i)} \in [0,1]$ is a smoothing parameter. The updates are performed backwards to the root node.

Details of the OCBA tree policy are shown in Algorithm \ref{alg:MCTS} to \ref{alg:backpropagate}.\par

\begin{algorithm}
	\DontPrintSemicolon
	\KwIn{Simulation budget (roll-out number) \(N\), root state node \(\bx_0\)}
	\KwOut{\(\hat{a}_{\bx_0}^*\), \(\hat{V}^*(\bx_0)\)}
	Set simulation counter \({n}\leftarrow 0\)\;
	\While{\(n < N \)}
	{
		$path \leftarrow selection(x_0)$ \;
		$leaf \leftarrow path[end]$\;
		$ r \leftarrow simulate(leaf)$\;
		$backpropagate(path, r)$\;
		$n \leftarrow n+1$
	}
	return action \(\hat{a}_{\bx_0}^* = \arg \max_{a \in A} \bar{Q}(\bx_0, a)\) \;
	\caption{MCTS}
	\label{alg:MCTS}
\end{algorithm}

\begin{algorithm}
	\DontPrintSemicolon
	\KwIn{root state node \(\bx_0\)}
	Sample a root-to-leaf path.\;
	 $path \leftarrow (~)$ \;
	$\bx \leftarrow \bx_0$\;
	\While{True}
	{
		Append state node $\bx$ to $path$\;
	    $N(\bx) \leftarrow N(\bx) + 1$\;
		\If{$\bx$ is a terminal node}
		{
		    return $path$\;
		}
		
		\eIf{$\bx$ is expandable}
		{
		    $\hat{a} \leftarrow expand(\bx)$\;
		    $\by \leftarrow expand((\bx, \hat{a}))$\;
		    Append state-action node $(\bx, \hat{a})$ and leaf state node $\by$ to $path$\;
	        $N(\bx, a) \leftarrow N(\bx, a) + 1$\;
	        $N(\bx) \leftarrow N(\bx) + 1$\;
		    return $path$\;
		}
		{
		    $\hat{a} \leftarrow OCBAselection(\bx)$\;
		    Append state-action node $(\bx,\hat{a})$ to $path$\;
	        $N(\bx, \hat{a}) \leftarrow N(\bx, a) + 1$\;
		    $\bx \leftarrow expand((\bx,\hat{a}))$\;
		}
	}
	\caption{\(selection(\bx_0)\)}
	\label{alg:selection}
\end{algorithm}

\begin{algorithm}
	\DontPrintSemicolon
	\KwIn{state node \(\bx\)}
    Identify \(\hat{a}_{\bx}^*= \arg\max_a \bar{Q}(\bx, a)\)\;
	\(\delta_{\bx}(\hat{a}_{\bx}^*,a)  \leftarrow \bar{Q}(\bx,\hat{a}_{\bx}^*)-\bar{Q}(\bx, a)\)\;
	Compute new sampling allocation
	\((\tilde{N}(\bx, a_1),\tilde{N}(\bx, a_2),\dots,\tilde{N}(\bx,a_{|A|}))\)\\
	by solving \Crefrange{eq:budget_allocation1}{eq:budget_allocation3}\;
	\(\hat{a} \leftarrow \arg\max_{a \in A}(\tilde{N}(\bx, a) - N(\bx, a) )\)\;
	return $\hat{a}$\;
	\caption{\(OCBASelection(\bx)\)}
	\label{alg:OCBAselection}
\end{algorithm}

\begin{algorithm}
	\DontPrintSemicolon
	\KwIn{a state node $\bx$ {\it or} a state-action node $(\bx, a)$}
	\KwOut{child node to be added to the tree}
	\eIf{the input node is a state node $\bx$}
	{
	    $S \leftarrow$ \{feasible actions of state $x$ that has been sampled less than $n_0$ times\}\; 
	    $\hat{a} \leftarrow $ random choice of $S$\;
	    Add $(\bx, \hat{a})$ to the tree if it is unvisited\;
	    return $\hat{a}$\;
	}
	{
	    Sample node $(\bx, a)$ at state node $\bx$ and obtain the child state node \(\by \sim P(\bx, a)(\cdot)\)\;
	    Add $\by$ to the tree if it is unvisited\;
	    return $\by$.
	    
	}
	\caption{\(expand(\bx~ or ~(\bx, a))\)}
	\label{alg:expand}
\end{algorithm}

\begin{algorithm}
	\DontPrintSemicolon
	\KwIn{state node \(\bx\)}
	$r \leftarrow 0$\;
	\While{True}
	{
	    \eIf{$\bx$ is not terminal}
	    {
	        find a random child state-action node $(\bx, a)$ of $\bx$\;
	        $r \leftarrow r + R(\bx, a)$\;
	        sample $a$ and obtain the child state node \(\by \sim P(\bx, a)(\cdot)\)\;
	        $\bx \leftarrow \by$\;
	    }
	    {
	        return $r$\;
	    }
	}
	\caption{\(simulate(\bx)\)}
	\label{alg:simulate}
\end{algorithm}

\begin{algorithm}
	\DontPrintSemicolon
	\KwIn{path to a leaf node \(path\), simulated reward $reward$}
	\For{node in reversed(path)}
	{
	   Update node values through \Crefrange{eq:bp_leaf}{eq:bp4}.
	}
	\caption{\(backpropagate(path, reward)\)}
	\label{alg:backpropagate}
\end{algorithm}

There are a few points worth emphasizing in Algorithm \ref{alg:OCBAselection}. First, $\tilde{N}(\bx, a_i)$ is the total number of samples for each action $i$ after the allocation. Given present information, i.e., all samples state node \(\bx\), OCBA-MCTS assumes now a total number of \(\sum_{a \in A} N(\bx, a) + 1\) samples available. By solving \Crefrange{eq:budget_allocation1}{eq:budget_allocation3}, the new budget allocation 
\((\tilde{N}(\bx, a_1),\tilde{N}(\bx, a_2),\dots,\tilde{N}(\bx,a_{|A|}))\) 
that maximizes APCS is calculated. Afterwards, one action based on \Cref{eq:select_policy} is selected to sample and move to the next stage. This ``most-starving'' implementation of the OCBA policy as given in Algorithm \ref{alg:OCBAselection} is fully sequential, as each iteration allocates only one sample to an action before the allocation decision is recomputed. It is also possible to allocate the sampling budget in a batch of size $\Delta>1$. We use the ``most-starving'' scheme, because it has been shown to be more efficient than the batch sampling scheme \cite{chen2006efficient}. However, the benefit of sampling in batches for MCTS is that in one iteration, multiple root-to-leaf paths can be examined, enabling parallelization of the algorithm. We will consider this in future research. \par 

Second, updating $\hat{V}(\bx_i)$ involves two stages: updating the value estimate along the path (\Cref{eq:bp3}) and taking the maximum over the values of the child state-action nodes (canonical way to update). Then the two values are mixed through $\alpha_{N(\bx_i)}$ to update $\hat{V}(\bx_i)$, as prior research (e.g., \cite{jiang2017monte,coulom2007efficient}) suggests mixing with $\alpha_{N(\bx_i)} \rightarrow 1$ (i.e., asymptotically achieves Bellman update) ensures more stable updates.\par 

Finally, although we present our algorithm in the context of solving an MDP, it can be applied to other tree structures such as MIN-MAX game trees or more general game trees, by setting the reward function and the \(\max\) and \(\min\) operators accordingly. \par 

\section{Analysis of OCBA-MCTS}\label{sec:analysis}
In this section, we first analyze how the OCBA tree policy in OCBA-MCTS balances exploration and exploitation mathematically.  Then, we present several theoretical results regarding OCBA-MCTS. The proofs are given in the appendix.

\subsection{Exploration-exploitation balance}
\Crefrange{eq:budget_allocation1}{eq:budget_allocation3} determine the new sampling budget allocation. First, \cref{eq:budget_allocation1} shows that the sub-optimal state-action nodes should be sampled proportional to their variances and inversely proportional to the squared differences between their sample means and that of the optimal state-action node. This represents a different type of trade off between exploration (sampling actions with high variances) and exploitation (sampling actions with higher sample means) compared to bandit-based algorithms.\par 


\subsection{Convergence analysis}
In this part, we present three theorems regarding OCBA-MCTS. The first theorem ensures the estimate of the value-to-go function converges to the true value. The second theorem proves that OCBA-MCTS will select the correct action, i.e., the PCS converges to 1. The last theorem guarantees that the APCS, which is a lower bound of PCS, is maximized by solving \Cref{eq:budget_allocation1,eq:budget_allocation2} in each step. It is shown that at each point of the tree policy when a decision needs to be made, the action that maximizes the APCS will be selected and sampled. Therefore, the OCBA tree policy gradually maximizes the overall APCS at the root, which is a lower bound for PCS.
\begin{theorem} [Asymptotic consistency]\label{thm:asym_unbias}
	Assume the expected cumulative reward at state-action node $(\bx, a)$ is a normal random variable with mean \(\mu(\bx, a)\) and variance \(\sigma^2(\bx, a)<\infty\), i.e., \(\hat{Q}(\bx, a) \sim N(\mu(\bx, a), \sigma^2(\bx, a) )\) for \(0 \le i < H\). Further assume \(\mu(\bx, a)\) is also normally distributed with unknown mean and known variance. Suppose the proposed OCBA-MCTS algorithm is run with a sampling budget \(N\) at root state node \(\bx_0\). Then at any subsequent nodes $\bx$,
	\begin{align*}
	\lim_{N\rightarrow\infty} \bar{Q}(\bx, a) &= \mathbb{E}[\hat{Q}(\bx, a)] = Q(\bx, a) ,\\
	\lim_{N\rightarrow\infty} \hat{V}(\bx) &= V^*(\bx),~ \forall ~ \bx\in \bX, ~ (\bx, a)\in \bX \times A_{\bx}.
	\end{align*}
\end{theorem}

\begin{theorem}[Asymptotic correctness]\label{thm:asym_correct}
	Under the same assumptions of \Cref{thm:asym_unbias}, the PCS converges to 1 for any state node $\bx \in \bX$, i.e.,
	\begin{align*}
	P &\bigg[\bigcap_{a\in A_{\bx}, a\ne \hat{a}_{\bx}^*} (\lim_{N\rightarrow\infty}\tilde{Q}(\bx,\hat{a}_{\bx}^*)-\lim_{N\rightarrow\infty}\tilde{Q}(\bx, a)) \ge 0\bigg] = 1, \\ 
	&\forall ~ \bx\in \bX,
	\end{align*}	
	where \(\hat{a}_{\bx}^* = \arg \max_{a \in A_{\bx}} \bar{Q}(\bx, a)\).
\end{theorem}

\begin{theorem}\label{thm:max_PCS}
	Under the same assumptions of \Cref{thm:asym_unbias}, the APCS defined in \Cref{eq:APCS} is maximized asymptotically with simulation budget allocation \((\tilde{N}(\bx, a_1),\tilde{N}(\bx, a_2),\dots,\tilde{N}(\bx,a_{|A|}))\) by solving \Cref{eq:budget_allocation1,eq:budget_allocation2} with total budget \(N\), i.e.,
	$	\sum_{a \in A_{\bx}}\tilde{N}(\bx, a) = N.	$	
\end{theorem}
\Cref{thm:max_PCS}, which follows from the result originally derived in \cite{chen2000simulation}, shows that at each point of the algorithm when a decision needs to be made, the action that maximizes the APCS will be selected and sampled. Therefore, the OCBA tree policy gradually maximizes the overall APCS at the root, which is a lower bound for PCS.

\subsection{Performance lower bound}
We take advantage of the normal distribution assumptions on the $Q$ functions and provide a lower bound on PCS.
\begin{theorem}[Lower bound on the probability of correct selection]\label{thm:PCS_bound}
	Under the same assumptions of \Cref{thm:asym_unbias}, the PCS at each stage and state is lower bounded by
	{\small
	\begin{align*}
	&PCS 
	\ge 1- \\
	&\sum_{a\in A_{\bx}, a\ne \hat{a}_{\bx}^*} \Phi \bigg(-\frac{\delta_{\bx}(\hat{a}_{\bx}^*,a)  \sqrt{N(\bx, \hat{a}_{\bx}^*) }}{ \sqrt{ \sigma^2(\bx, \hat{a}_{\bx}^*) + 
			\sigma(\bx, \hat{a}_{\bx}^*) \sigma^2(\bx, a) \sum_{\tilde{a} \in A_{\bx}, \tilde{a} \ne \hat{a}^*_{\bx}}  \frac{ r(\tilde{a}, a )  }{\sigma(\bx, \tilde{a})}   } } \bigg),
	\end{align*}
	where \(\Phi(\cdot)\) is the cdf of standard normal distribution and 
	\begin{align*}
	r_{\bx}(\tilde{a}, a) = \frac{\sigma(\bx, \tilde{a}) \delta_{\bx}(\hat{a}_{\bx}^*,a)^2 }{\sigma(\bx, a) \delta_{\bx}(\hat{a}_{\bx}^*,\tilde{a})  }.
	\end{align*}
	}
\end{theorem}

\section{Numerical examples}\label{sec:simulation}
In this section, we evaluate our proposed OCBA-MCTS on two tree search problems against the well-known UCT \cite{kocsis2006bandit}. The effectiveness is measured by PCS, which is estimated by the fraction of times the algorithm chooses the true optimal action. We first evaluate our algorithm on an inventory control problem with random non-normal reward. Then we apply our algorithm to the game of tic-tac-toe. \par 

For convenience, we restate the UCT tree policy here. At a state node $\bx$, the UCT policy will select the child state-action node with the highest upper confidence bound, i.e.,
\begin{align}\label{eq:ucb}
    \hat{a} 
    =& \arg \max_{a \in A_x} \big\{ \bar{Q}(\bx, a) + w_e \sqrt{\frac{2\log  \sum_{a' \in A_x}N(\bx, a')}{ N(\bx, a)}}  \big\},
\end{align}
where $w_e$ is the ``exploration weight". The original UCT algorithm assumes the value function in each stage is bounded in $[0,1]$ because it sets $w_e = 1$, whereas the support is unknown in many practical problems. Therefore, in general, $w_e$ needs to be tuned to encourage exploration. \par 

For all experiments, we set the smoothing parameter in \Cref{eq:bp4} in the backpropagation phase to $\alpha_{N(\bx)} = 1- \frac{1}{5N(\bx)}$. 
Since initial estimates of sample variance can be less accurate with small $n_0$, we add an initial variance $\sigma_0^2 > 0$, which decays as the number of visits grows, to the sample variance to encourage exploration. Specifically, we set  
\begin{align*}
\hat{\sigma}^2(\bx, a) &=  \frac{1}{N(\bx, a)} \sum_{t = 1}^{N(\bx, a)}\big(\hat{Q}^t(\bx, a) - \bar{Q}(\bx, a)\big)^2 + \sigma_0^2/N(\bx, a),
\end{align*}
where the first term is the sample variance, and second term vanishes as $N(\bx, a)$ grows.  \par

\subsection{Inventory control problem}
We now evaluate the performance of OCBA-MCTS using the inventory control problem in \cite{chang2005adaptive}. The objective is to find the initial order quantity that minimizes the total cost over a finite horizon. At decision period $i$, we denote by \(D_i\) the random demand in period \(i\), $\bx_i = (x_i, i)$ the state node, where \(x_i\) is the inventory
level at the end of period \(i\) (which is also the inventory at the beginning of period \(i+1\)), $(\bx_i, a_i)$ the corresponding child state-action node with \(a_i\) being the order amount in period \(i\), \(p\) the per period per unit demand lost penalty cost, \(h\) the per period per unit inventory holding cost, \(K\) the fixed (set-up) cost per order, \(M\) the maximum inventory level (storage capacity) and \(H\) the number of simulation stages. We set \(M = 20\), initial state \(x_0 = 5\), \(h = 1\), \(H = 3\), \(D_i \sim DU(0,9)\) (discrete uniform, inclusive), and consider two different settings for $p$ and $K$:
\begin{enumerate}
    \item Experiment 1: $p = 10$ and $K=0$;
    \item Experiment 2: $p = 1$ and $K=5$.
\end{enumerate}
The reward function, which in this case is the negative of the inventory cost in stage \(i\), is defined by 
\begin{align*}
R(\bx_i,a_i) =& -(h\max\{0,x_i+a_i-D_i\} + \\ 
&p\max\{0,D_i-x_i-a_i\} 
+ K\mathds{1}_{\{a_i>0\}}),
\end{align*} 
where \(\mathds{1}\) is the indicator function, and the state transition follows
\begin{align*}
x_{i+1} = \max(0, x_i+a_i-D_i),
\end{align*}
where 
\begin{align*}
    a_i \in A_{x_i} = \{a| x_i+a \le M  \}.
\end{align*}
For UCT, to accommodate the reward support not being \([0,1]\), we adjust the exploration weight when updating a state-action node, i.e., set $w_e$ initially to 1, then in the backpropagation step, update $w_e$ by 
\begin{align*}
    w_e = \max(w_e, |\hat{Q}^{N(\bx, a)}(\bx, a)|),
\end{align*}
where $\hat{Q}^{N(\bx, a)}(\bx, a)$ is obtained in \Cref{eq:bp1}. The initial variance $\sigma_0^2$ is set to 100.
For both OCBA-MCTS and UCT, we set the number of expansions ($n_0$) to 4 for depth 1 state-action nodes (i.e., the child nodes of the root) and to 2 for all other state action nodes in Experiment 1, and set $n_0$ to 2 for all nodes in Experiment 2. The different values of $n_0$ are due to the variance decreasing with the depth of a node, and Experiment 2 is a relatively easier problem.

For both experiment settings, each algorithm is repeated $1,000$ times at each simulation budget level $N$ to estimate PCS.  Since Experiment 1 is a much harder problem compared to Experiment 2, more rollouts (budget) are required. Therefore, $N$ ranges from $10,000$ to $20,000$ and from $50$ to $170$ for Experiments 1 and 2, respectively. \par 

\begin{figure}[h]	
    \captionsetup[subfigure]{aboveskip=-1pt,belowskip=8pt}
	\centering
	\begin{subfigure}{0.5\textwidth}
	\centering
		\includegraphics[scale = 0.5]{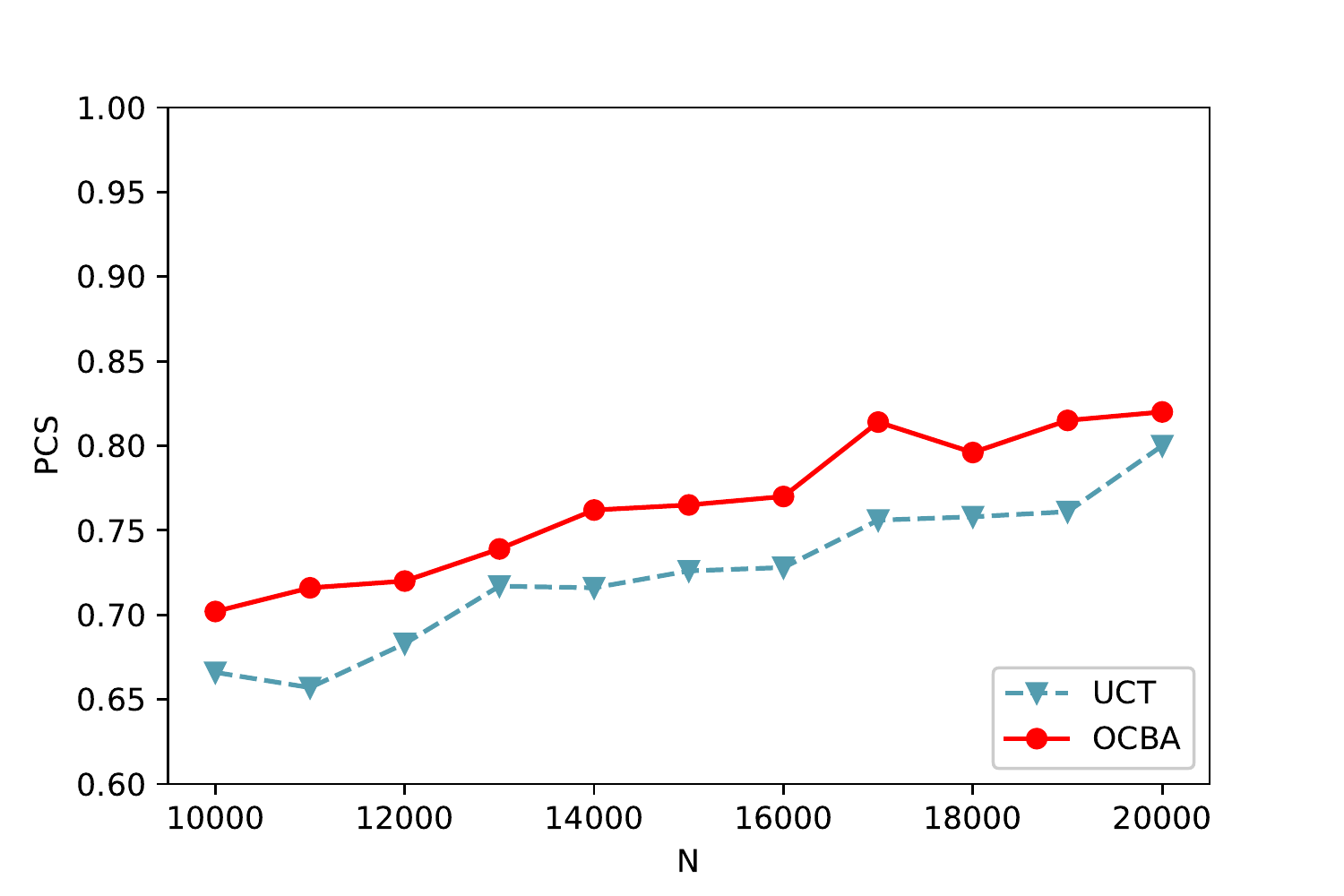}
		\caption{Experiment 1: $p = 10$, $K=0$}
		\label{fig:inventory_pcs_p10_K0}
	\end{subfigure}
	~
	\begin{subfigure}{0.5\textwidth}
	\centering
		\includegraphics[scale = 0.5]{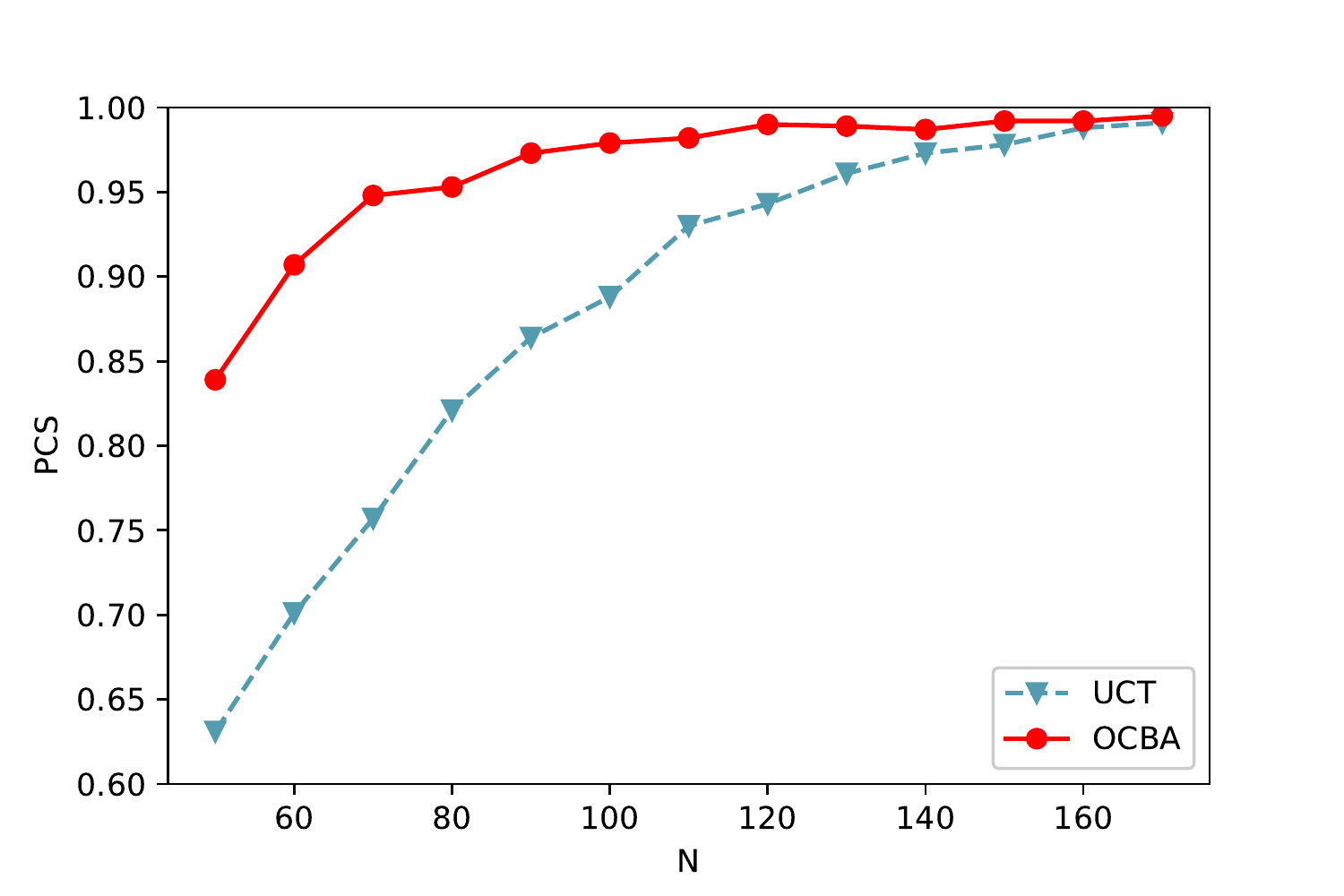}
		\caption{Experiment 2: $p = 1$, $K=5$}
		\label{fig:inventory_pcs_p1_K5}
	\end{subfigure}
	\caption{The estimated PCS as a function of sampling budget achieved by UCT-MCTS and OCBA-MCTS for inventory control problem, averaged over 1,000 runs.}
	\label{fig:inventory_pcs}
\end{figure}
The estimated PCS curves for both experiments are illustrated in \Cref{fig:inventory_pcs}, where the standard error ($= \sqrt{PCS(1-PCS)/N}$) is small and thus omitted for clarity.
OCBA-MCTS achieves better PCS for both experiment setups. For Experiment 1 (optimal action \(a^*_0 = 4\)), as shown in \Cref{fig:inventory_pcs_p10_K0}, OCBA-MCTS achieves a $10\%$ higher PCS (absolute) compared to UCT. For Experiment 2 (optimal action \(a^*_0 = 0\)), we see a $20\%$ performance gap between UCT and OCBA-MCTS when the number of samples is less than 100, after which UCT gradually closes the gap as expected.


It is also beneficial to compare the distribution of budget allocation of OCBA-MCTS and UCT to show the exploration-exploitation balance of OCBA-MCTS. For convenience, we label the child actions of the root node from $0$ to $15$, where action $i$ denotes ordering $i$ units. Figures \ref{fig:inventory_sample_dist_K0_p10} and \ref{fig:inventory_sample_dist_K5_p1} illustrate the average number of visits, average estimated value function, and average estimated standard deviation of all child state-action nodes of the root node over 1,000 repeated runs with 20,000 and 170 rollouts for Experiment 1 and Experiment 2, respectively. Note that although the estimated standard deviation does not play a role in determining the allocation for UCT, we still plot it for reference. Both figures show that the number of visits to children nodes is, to some extent, proportional to the estimated value of the node for UCT.
On the other hand, OCBA-MCTS puts more effort on the estimated optimal and second optimal actions (actions 4 and 3 for Experiment 1 and actions 0 and 1 for Experiment 2, respectively), as illustrated in Figures \ref{fig:inventory_sample_dist_K0_p10_ocba} and \ref{fig:inventory_sample_dist_K5_p1_ocba}. \par 

In Experiment 1 where there are two competing actions with similar estimated values (actions 3 and 4, with action 4 being the optimal), OCBA-MCTS will spend most of its sampling budget on those two potential actions and put much lesser effort on clearly inferior actions, such as actions 6 to 14, compared to UCT. This strategy makes more sense when the objective is to identify the best action, and thus is more suitable for MCTS problems, as the ultimate goal is to make a decision. It is also interesting to note that OCBA-MCTS actually allocates slightly more visits to the competing suboptimal action than the optimal one (mean $8486$ and $8468$ for actions 3 and 4, respectively), which will not happen in bandit-based policies, as their goal is to minimize regret, and thus will put more effort on exploiting the estimated optimal action. In Experiment 2 where the optimum is slightly easier to find, although OCBA-MCTS allocates a larger fraction of samples to suboptimal actions compared to that in Experiment 1, most of the samples are still allocated to the top 2 actions as shown in \Cref{fig:inventory_sample_dist_K5_p1_ocba}, whereas UCT performs similar to that in Experiment 1.
\begin{figure}[ht]	
    \captionsetup[subfigure]{aboveskip=-1pt,belowskip=8pt}
	\centering
	\begin{subfigure}{0.5\textwidth}
		\includegraphics[scale = 0.6]{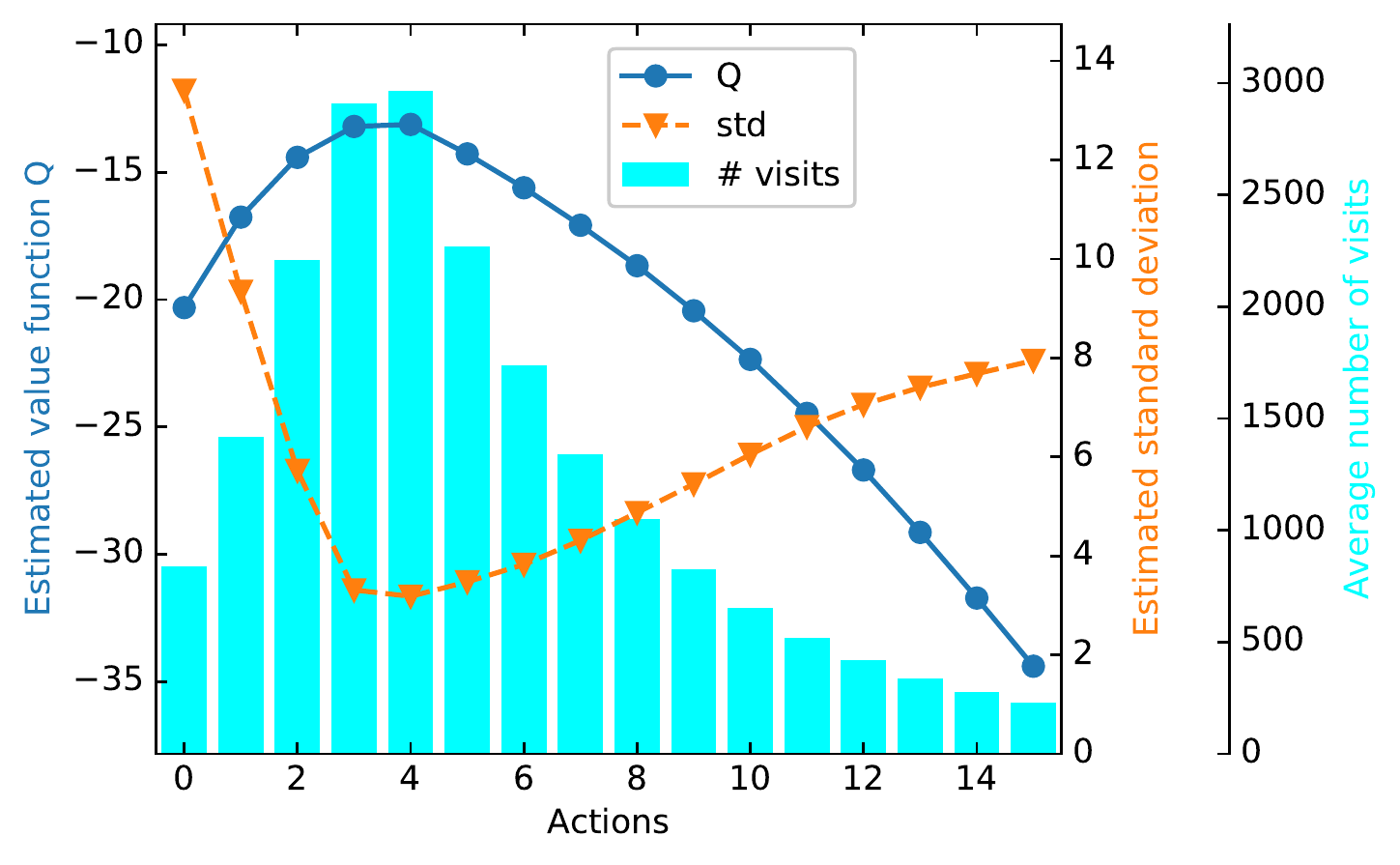}
		\caption{UCT}
		\label{fig:inventory_sample_dist_K0_p10_uct}
	\end{subfigure}
	~ ~
	\begin{subfigure}{0.5\textwidth}
		\includegraphics[scale = 0.6]{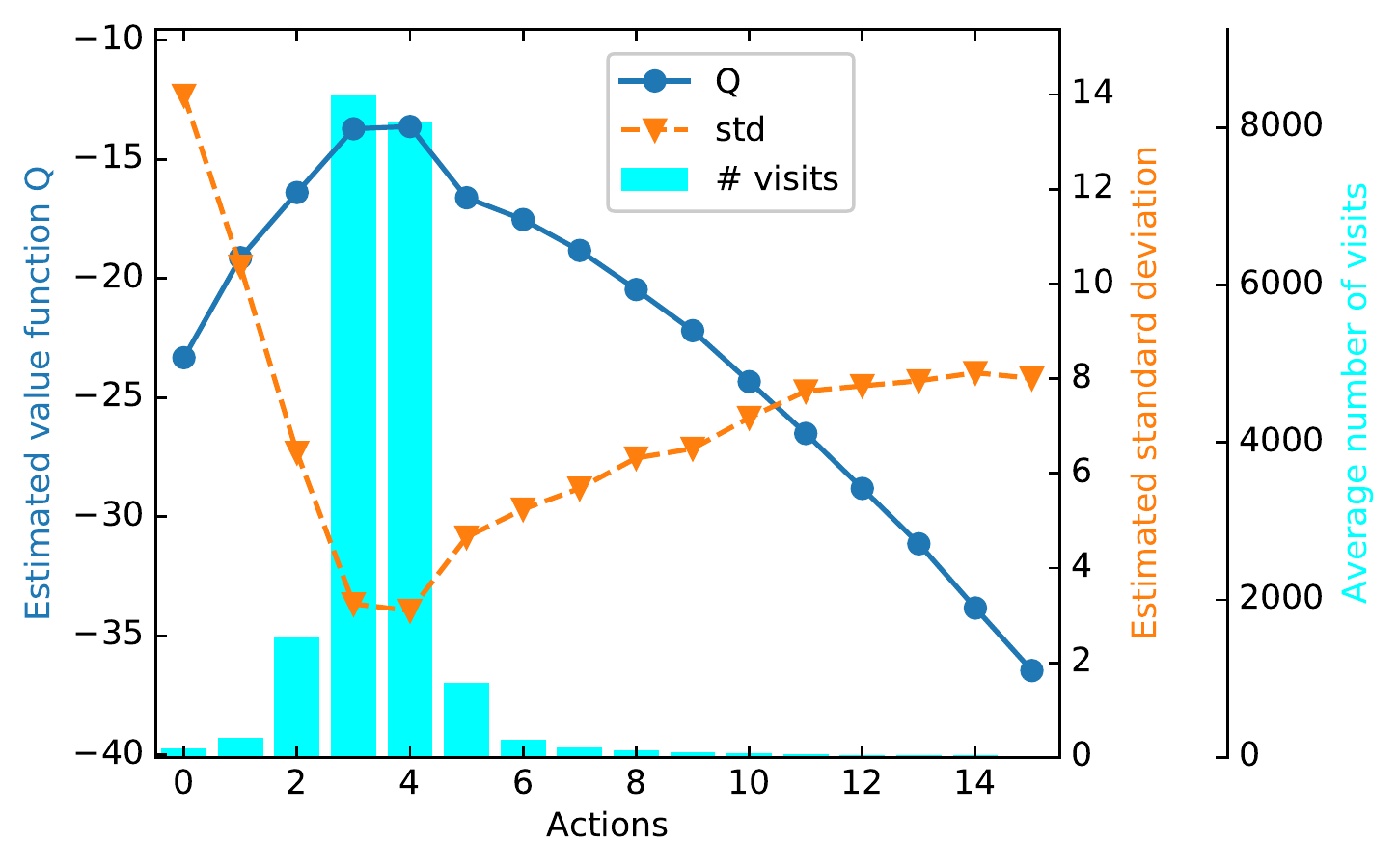}
		\caption{OCBA}
		\label{fig:inventory_sample_dist_K0_p10_ocba}
	\end{subfigure}
	\caption{Sampling distribution for Experiment 1 with $N=20,000$, averaged over 1,000 runs.}
	\label{fig:inventory_sample_dist_K0_p10}
\end{figure}

\begin{figure}[ht]	
    \captionsetup[subfigure]{aboveskip=-1pt,belowskip=8pt}
	\centering
	\begin{subfigure}{0.5\textwidth}
		\includegraphics[scale = 0.6]{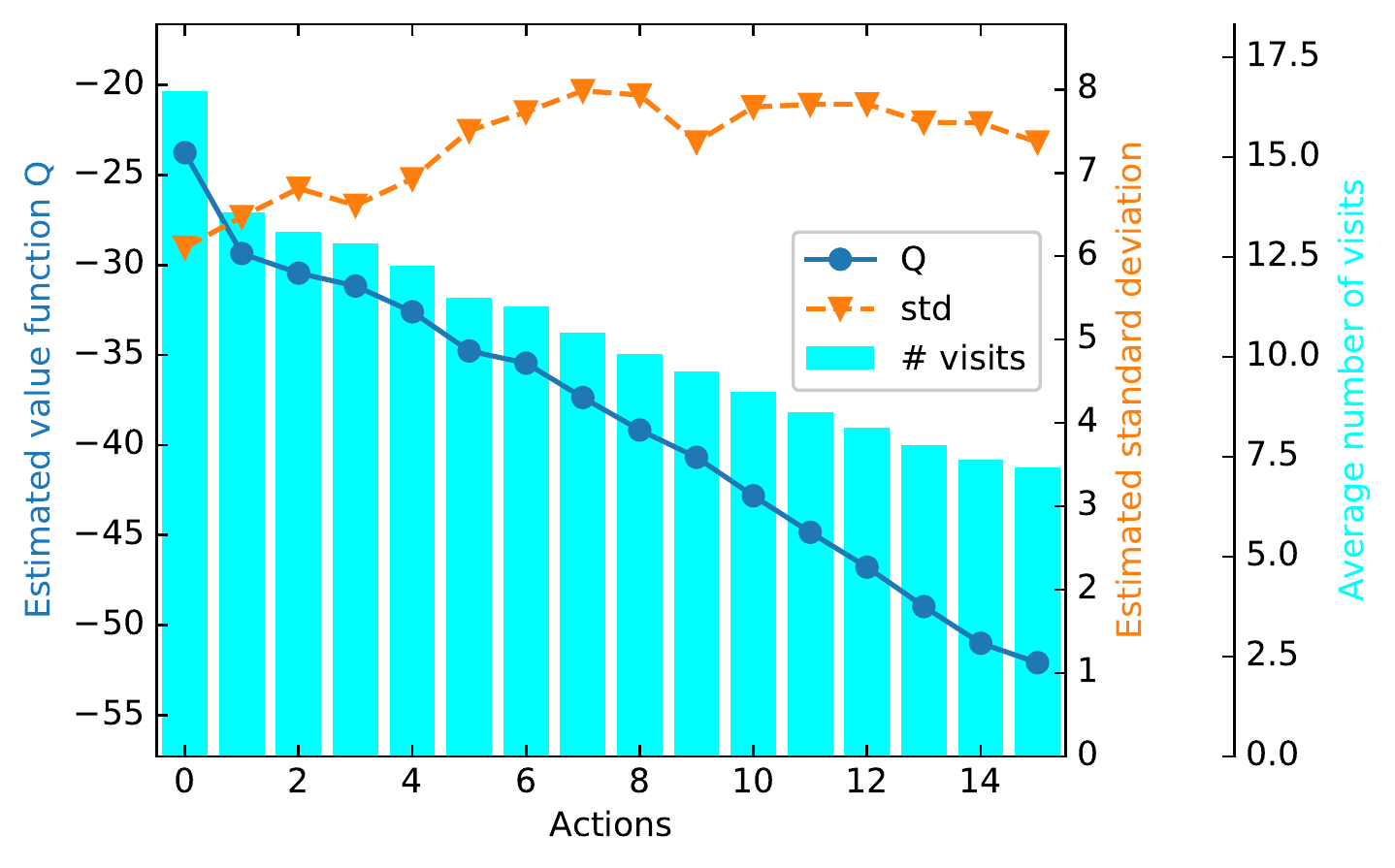}
		\caption{UCT}
		\label{fig:inventory_sample_dist_K5_p1_uct}
	\end{subfigure}
	~
	\begin{subfigure}{0.5\textwidth}
		\includegraphics[scale = 0.6]{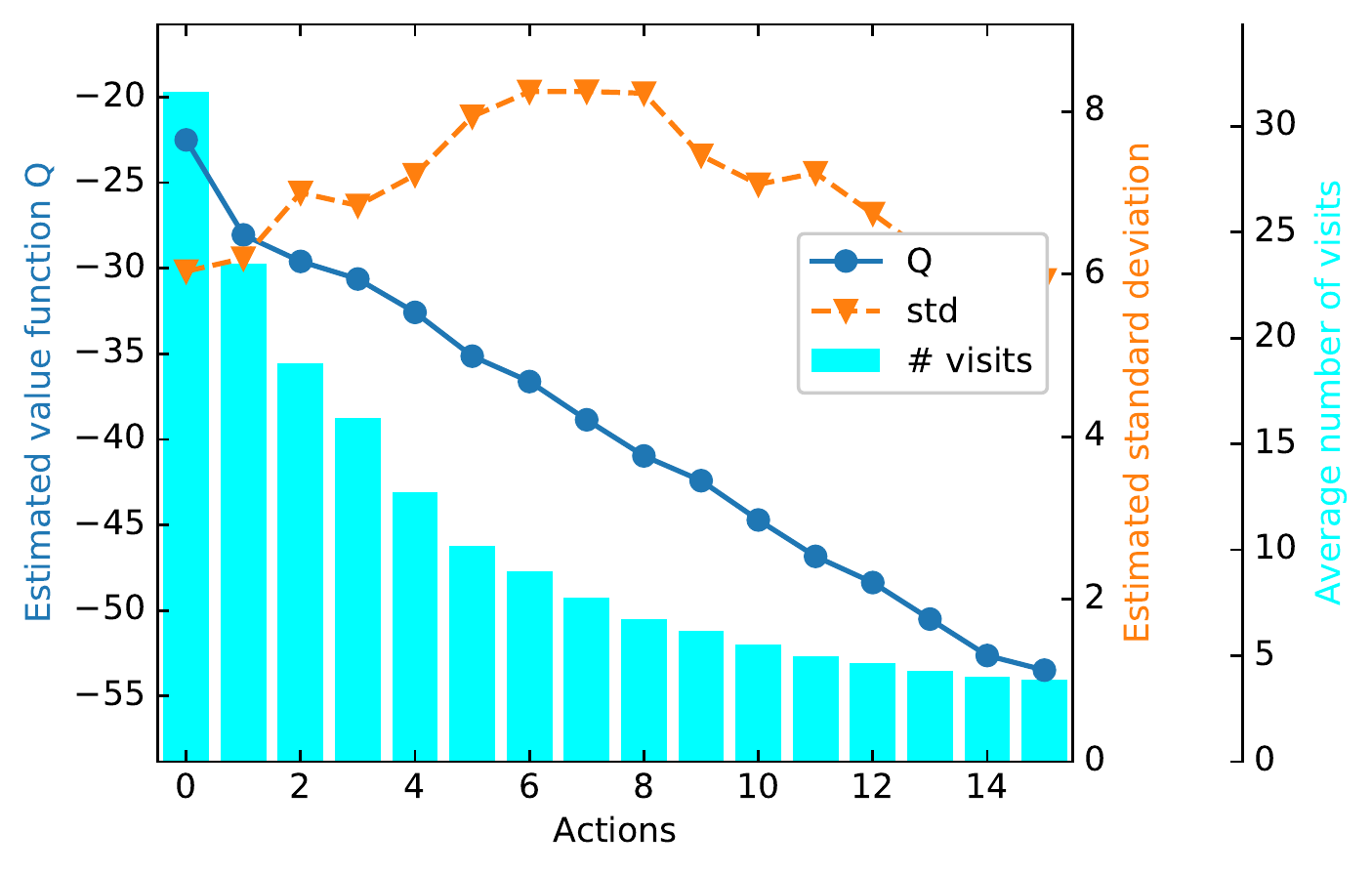}
		\caption{OCBA}
		\label{fig:inventory_sample_dist_K5_p1_ocba}
	\end{subfigure}
	\caption{Sampling distribution for Experiment 2 with $ N=170 $, averaged over 1,000 runs.}
	\label{fig:inventory_sample_dist_K5_p1}
\end{figure}

\subsection{Tic-tac-toe}
In this section, we apply OCBA-MCTS and UCT to the game of tic-tac-toe to identify the optimal move. Tic-tac-toe is a game for two players who take turns marking `X' (Player 1) and `O' (Player 2) on a $3\times 3$ board. The objective for Player 1 (Player 2) is to mark 3 consecutive `X' (`O') in a row, column or diagonal. If both players act optimally, the game will always end in a draw. \par 

For ease of presentation, we number the spaces sequentially as shown in \Cref{fig:tic_tac_toe_actions}. We use OCBA-MCTS and UCT to represent Player 2, with Player 1 marked `X' on space 0 as shown in \Cref{fig:tic_tac_toe_root}. In this situation, the optimal move for Player 2 will be marking space 4 (shown in \Cref{fig:tic_tac_toe_optimal}), as taking any other space will end up in losing the game if Player 1 plays optimally.
\begin{figure}[h]
    \centering
	\begin{subfigure}{0.15\textwidth}
	    \centering
		\includegraphics[scale = 0.5]{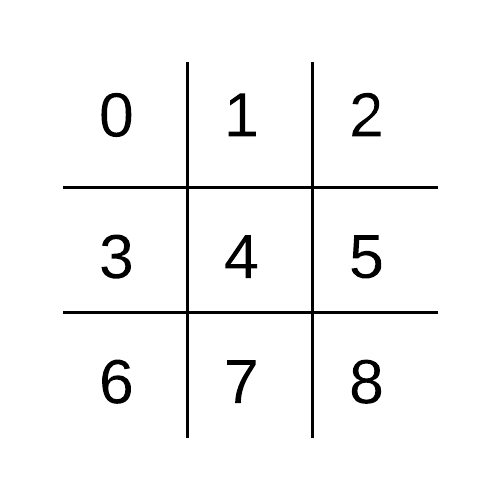}
		\caption{Action layout}
		\label{fig:tic_tac_toe_actions}
	\end{subfigure}
	\begin{subfigure}{0.15\textwidth}
	    \centering
		\includegraphics[scale = 0.5]{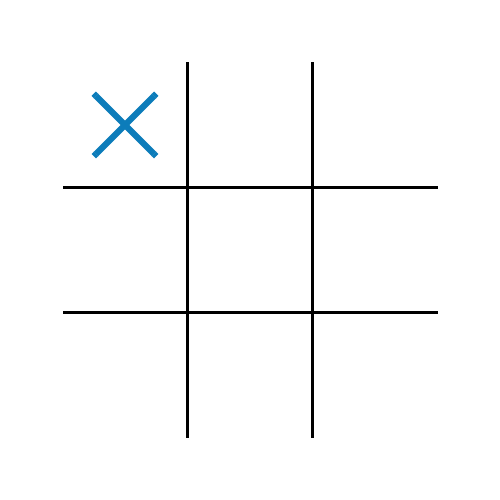}
		\caption{Root node}
		\label{fig:tic_tac_toe_root}
	\end{subfigure}
	\begin{subfigure}{0.15\textwidth}
	    \centering
		\includegraphics[scale = 0.5]{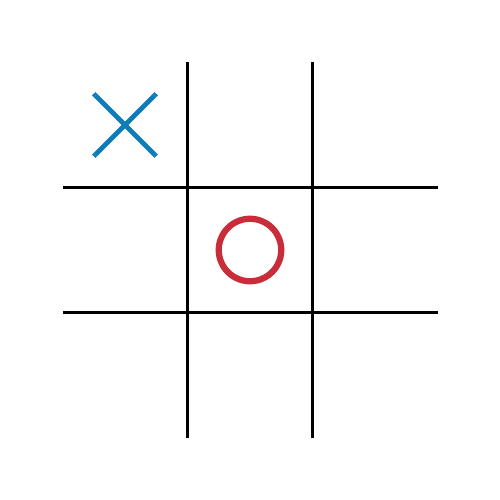}
		\caption{Optimal}
		\label{fig:tic_tac_toe_optimal}
	\end{subfigure}
	\caption{Tic-tac-toe board setup.}
	\label{fig:tic_tac_toe_setup}
\end{figure}
In this game, Player 2 (MCTS algorithm) makes decisions at even stages ($0, 2, 4, \dots $) and Player 1 makes decisions at odd stages ($1, 3, \dots $). The state transitioning is deterministic and Player 1's move is modeled using a randomized policy. We consider two different policies for Player 1:
\begin{enumerate}
    \item Experiment 3: Player 1 plays randomly, i.e., with equal probability to mark any feasible space;
    \item Experiment 4: Player 1 plays UCT.
\end{enumerate}
We compare the performance of OCBA-MCTS and UCT on Player 2 in both experiments.
At state node $\bx$, the reward function for taking action $a$ is defined according to the following rules: immediately after taking the action, if Player 2 wins the game, $R(\bx, a) = 1$, if it leads to a draw, $R(\bx, a) = 0.5$; otherwise (Player 2 loses or in any non-terminating state), $R(\bx, a) = 0$. $n_0$ is set to 2 across all nodes for both UCT and OCBA-MCTS.  
Since the value function for all state-action nodes is now bounded in $[0 , 1]$, we set $w_e = 1$ throughout the entire experiment for UCT policies. The initial variance $\sigma_0^2$ is set to 10. For Experiment 4 where Player 1 plays UCT, its goal is to {\it minimize} the reward, therefore, Player 1 will select the action that minimizes the lower confidence bound, i.e.,
\begin{align*}
        \hat{a} 
    =& \arg \min_{a \in A_x} \big\{ \bar{Q}(\bx, a) - w_e \sqrt{\frac{2\log  \sum_{a' \in A_x}N(\bx, a')}{ N(\bx, a)}}  \big\}.
\end{align*}


Similar to the previous section, we plot the PCS of the two algorithms as a function of the number of rollouts, which ranges from 300 to 700 for both experiments and the PCS is estimated over $2000$ independent experiments at each rollout level. The results are shown in \Cref{fig:tic_tac_toe_results}, which indicates that the proposed OCBA-MCTS produces a more accurate estimate of the optimal action compared to UCT.
\begin{figure}[ht]	
    \captionsetup[subfigure]{aboveskip=-1pt,belowskip=8pt}
	\centering
	\begin{subfigure}{0.5\textwidth}
	\centering
		\includegraphics[scale = 0.5]{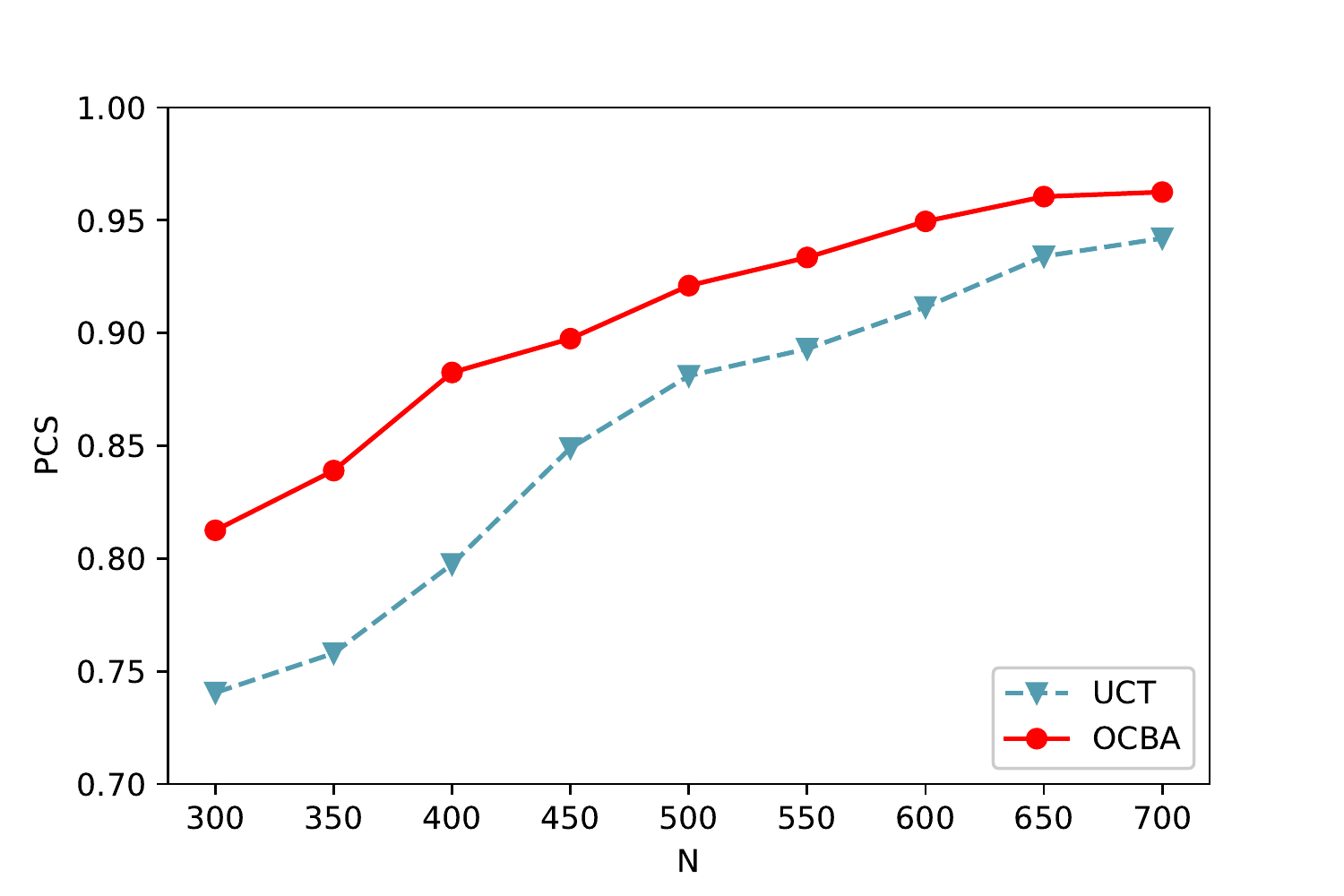}
		\caption{Experiment 3: Player 1 plays randomly.}
		\label{fig:tic_tac_toe_results_random_opponent}
	\end{subfigure}
	~
	\begin{subfigure}{0.5\textwidth}
	\centering
		\includegraphics[scale = 0.5]{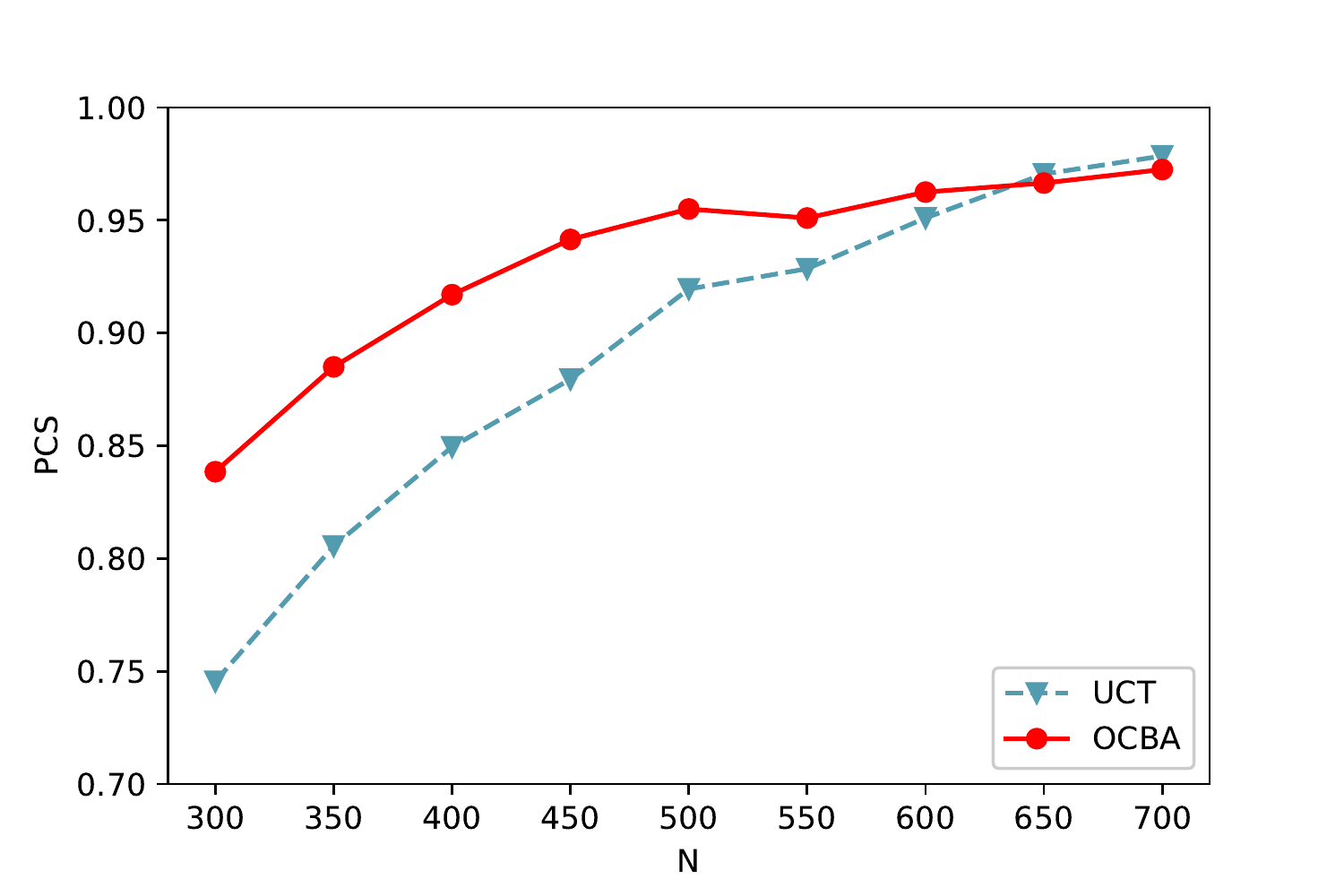}
		\caption{Experiment 4: Player 1 plays UCT.}
		\label{fig:tic_tac_toe_results_uct_opponent}
	\end{subfigure}
	\caption{The estimated PCS as a function of sampling budget achieved by UCT-MCTS and OCBA-MCTS for tic-tac-toe, averaged over 2000 runs.}
	\label{fig:tic_tac_toe_results}
\end{figure}
Both experiments show that OCBA-MCTS is better at finding the optimal move when the sampling budget is relatively low. The performance of UCT and OCBA-MCTS become comparable when more samples become available. We also note that there is a greater performance gap between UCT and OCBA-MCTS in Experiment 3 than in Experiment 4: in Experiment 3, OCBA-MCTS achieves $10\%$ better PCS, whereas in Experiment 4, the difference is around $5\%$ when $N<500$ and soon catches up as $N$ increases. This is expected, as it becomes easier to determine the optimal action when the opponent applies an AI algorithm (i.e., Player 1 has a better chance to take its optimal action). In this case, space 4 becomes a clear optimum and therefore Player 2's UCT algorithm tends to exploit it more, which leads to better performance.

The sampling distributions for OCBA-MCTS and UCT with $N = 700$ for both experiments are shown in Figures \ref{fig:TTT_sample_distribution_random_opponent} and \ref{fig:TTT_sample_distribution_uct_opponent}. In this game, since a relatively clear optimum is available, OCBA-MCTS and UCT behaved differently compared to that in the inventory control problem. As shown in Figures \ref{fig:TTT_sample_distribution_random_opponent_uct} and \ref{fig:TTT_sample_distribution_uct_opponent_uct}, 
UCT spends most of the sampling budget exploiting this action, whereas OCBA will still try to explore other suboptimal actions due to its tendency to better balance exploration and exploitation.
\begin{figure}[ht]	
    \captionsetup[subfigure]{aboveskip=-1pt,belowskip=8pt}
	\centering
	\begin{subfigure}{0.5\textwidth}
		\includegraphics[scale = 0.55]{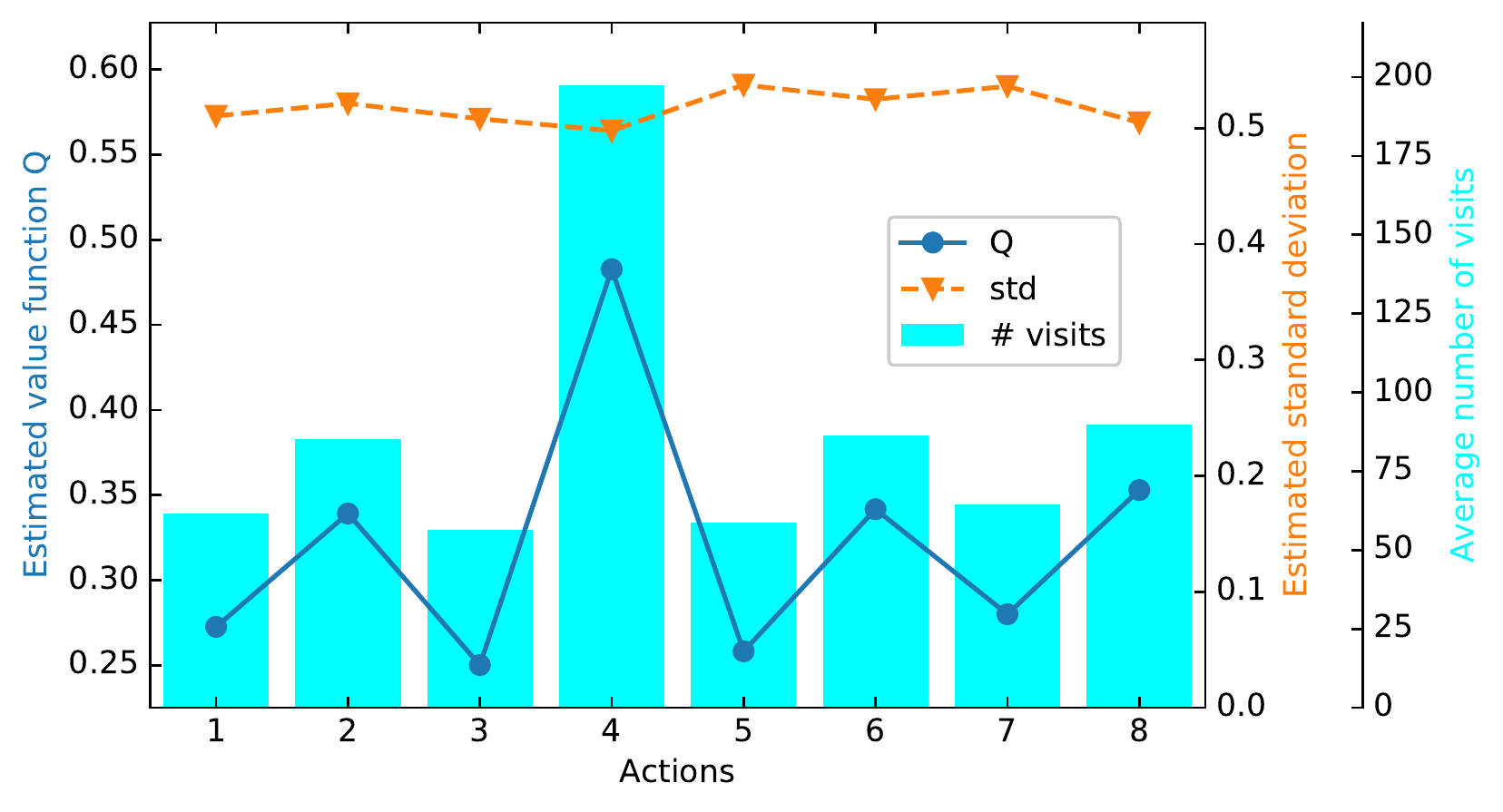}
		\caption{UCT}
		\label{fig:TTT_sample_distribution_random_opponent_uct}
	\end{subfigure}
	~
	\begin{subfigure}{0.5\textwidth}
		\includegraphics[scale = 0.55]{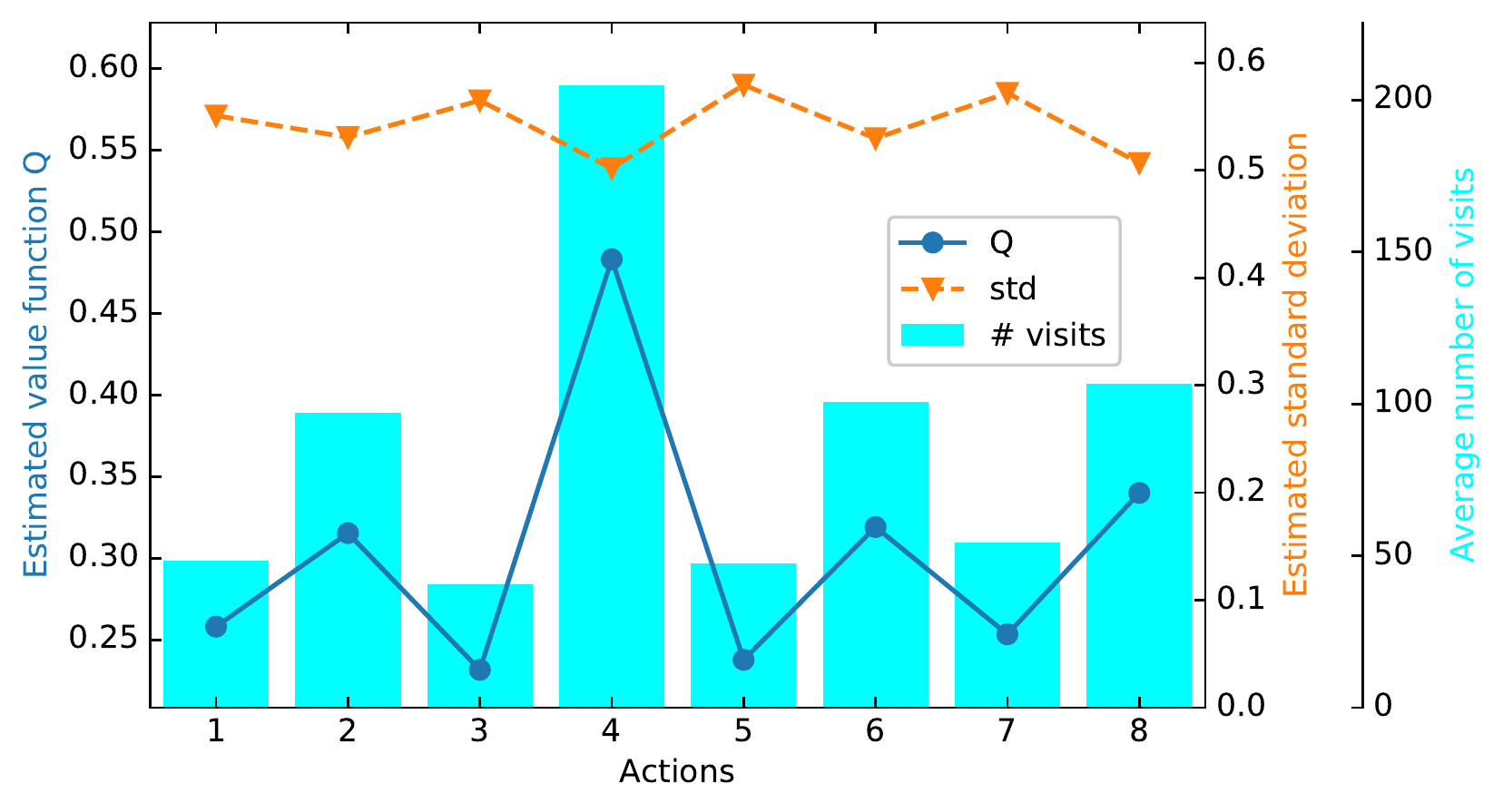}
		\caption{OCBA}
		\label{fig:TTT_sample_distribution_random_opponent_ocba}
	\end{subfigure}
	\caption{Sampling distributions for Experiment 3, averaged over 2000 runs.}
	\label{fig:TTT_sample_distribution_random_opponent}
\end{figure}

\begin{figure}[ht]	
    \captionsetup[subfigure]{aboveskip=-1pt,belowskip=8pt}
	\centering
	\begin{subfigure}{0.5\textwidth}
		\includegraphics[scale = 0.55]{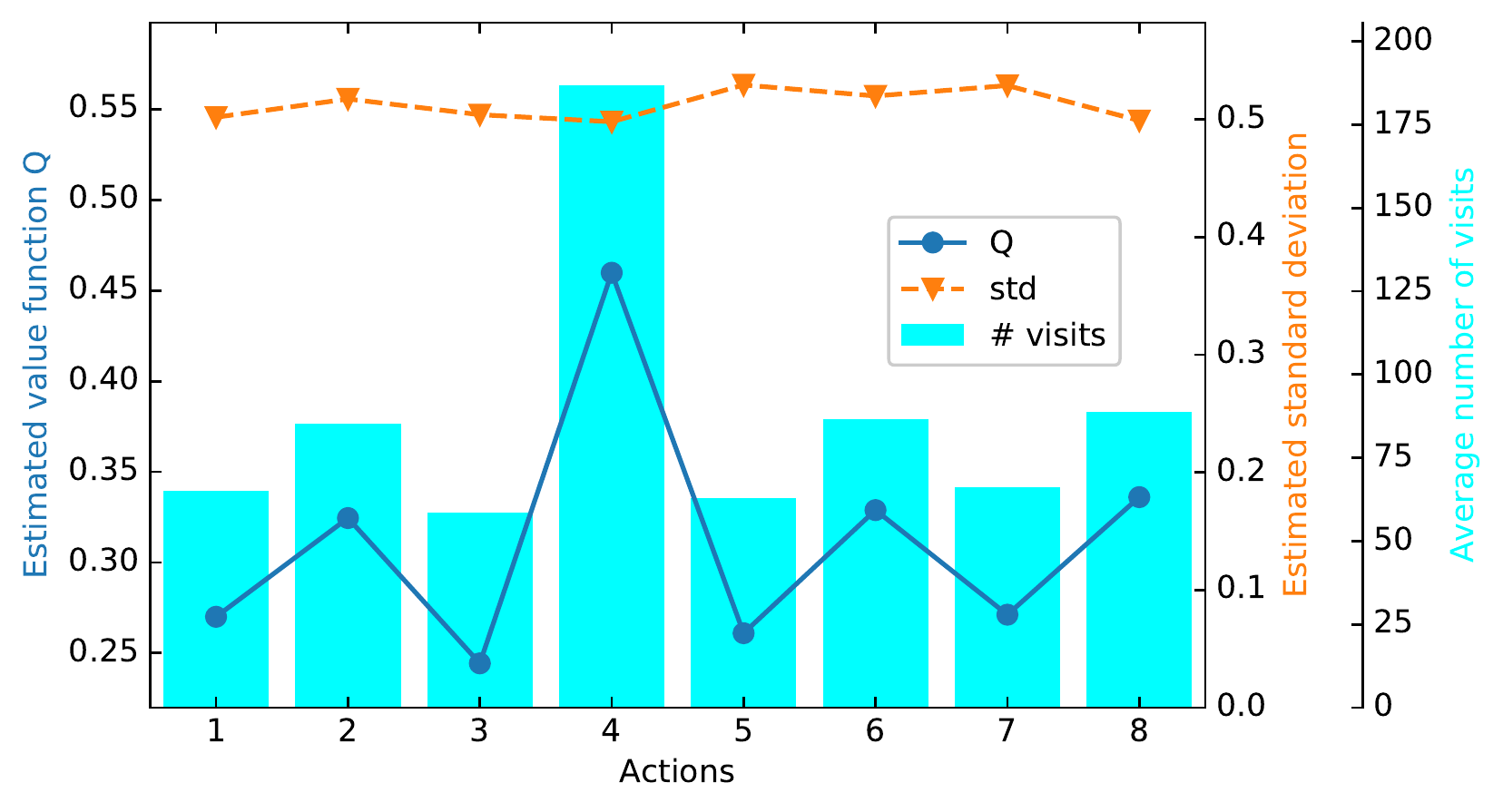}
		\caption{UCT}
		\label{fig:TTT_sample_distribution_uct_opponent_uct}
	\end{subfigure}
	~
	\begin{subfigure}{0.5\textwidth}
		\includegraphics[scale = 0.55]{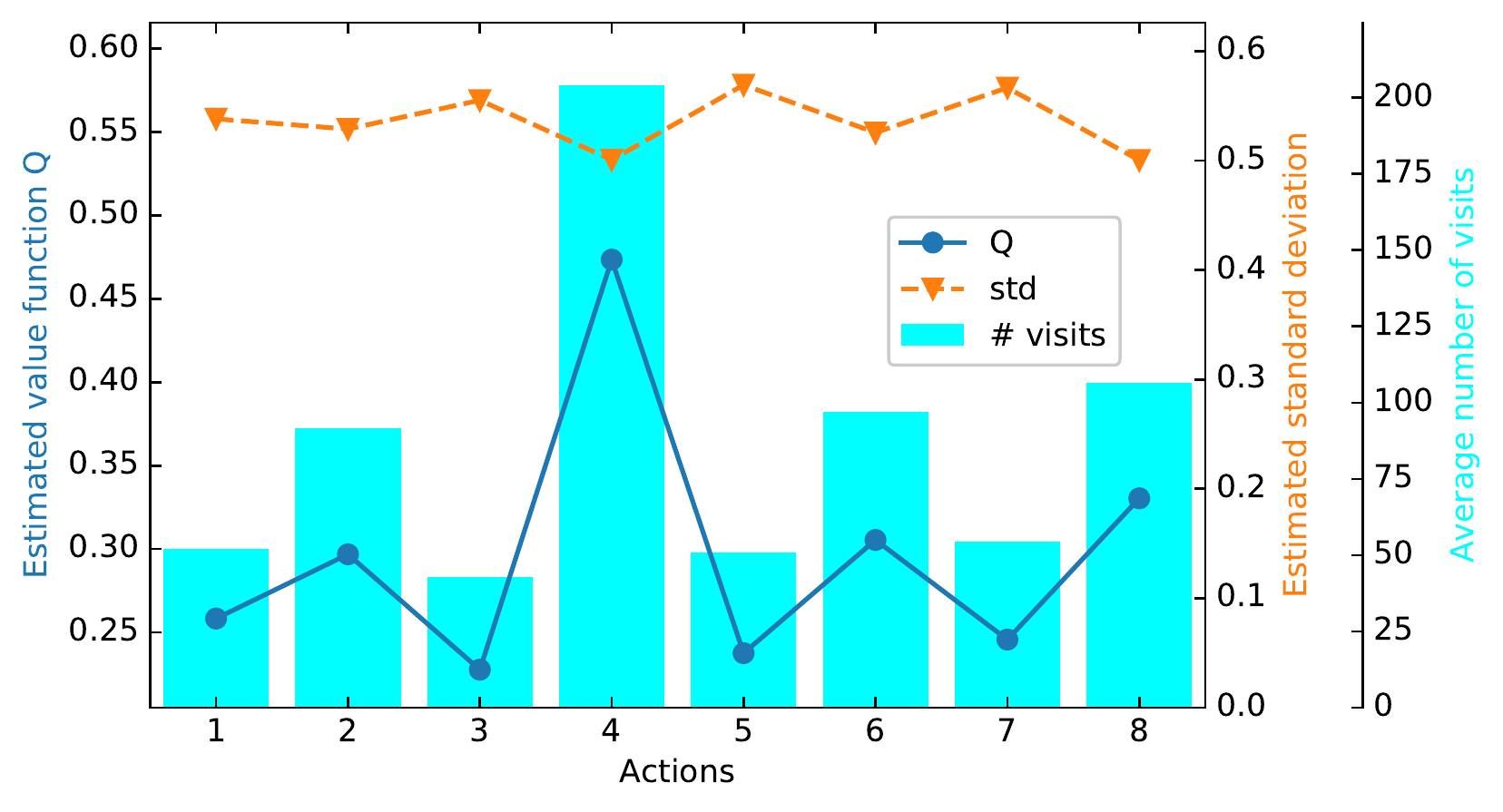}
		\caption{OCBA}
		\label{fig:TTT_sample_distribution_uct_opponent_ocba}
	\end{subfigure}
	\caption{Sampling distributions for Experiment 4, averaged over 2000 runs.}
	\label{fig:TTT_sample_distribution_uct_opponent}
\end{figure}

In summary, the proposed OCBA-MCTS outperforms UCT in both experiments in finding the optimal action at the root. Since the objective of the proposed OCBA tree policy is to maximize PCS, it leads to different budget allocation and better PCS.

\section{Conclusion and future research}\label{sec:conclusion}
In this paper, we present a new OCBA tree policy for MCTS. Unlike bandit-based tree policies (e.g., UCT), the new policy maximizes PCS at the root node, and in doing so, balances the exploration and exploitation trade-off differently. Furthermore, the new OCBA tree policy relaxes the assumption of known bounded support on the reward distribution, and thus makes MCTS more generally applicable.\par 

For future research, we intend to explore the use of a batch sampling scheme in Algorithm \ref{alg:selection}, which allocates a batch of \(\Delta > 1\) samples at each node. With batch sampling and updating, we may exploit the power of parallel computing to more quickly identify the optimal action.
Furthermore, the proposed OCBA-MCTS algorithm aims at selecting only the optimal action at the root. It may be of interest to show that our algorithm is \(\epsilon-\delta\)-correct and establish an upper bound on the sample complexity.
\bibliographystyle{ieeetr}

\begin{appendices}
\section{Convergence analysis}\label{app:sec:convergence_analysis}
To prove that our algorithm correctly selects the optimal action as the sampling budget goes to infinity, we first prove that at each stage, the PCS converges to 1. The process of our algorithm at each single stage is OCBA adapted from \cite{chen2000simulation}. OCBA tries to identify the alternative with highest mean from a set of normal random variables (alternatives) with means \(J_i\) and known variances \(\sigma_i^2\), \(i = 1,2,\dots, k\) by efficiently allocating samples that maximizes APCS. OCBA assumes that \(J_i\) is also normally distributed. Here we present OCBA again in Algorithm \ref{alg:OCBA} for convenience. The budget allocation process is similar to \Crefrange{eq:budget_allocation1}{eq:budget_allocation3}. First define
\begin{align*}
\bar{J}_i :=& \frac{1}{l_i}\sum_{m=1}^{l_i}\hat{J}^m_i,\\
b :=& \arg\max_i \bar{J}_i,\\
\delta(b,i) :=& \bar{J}_b - \bar{J}_i, ~\forall i \ne b,
\end{align*}	
where \(l_i\) is the number of samples for alternative \(i\), \(\hat{J}^m_i\) is the \(m\)-th sample of \(J_i\) for \(1\le i\le k\), \(1\le m \le l_i  \).
The new allocations \((\tilde{l}_1,\tilde{l}_2,\dots,\tilde{l}_{k})\) with budget \(T > \sum_i l_i\) can be obtained by solving the set of equations:
\begin{align}
\frac{\tilde{l}_i}{\tilde{l}_j} &= \big(\frac{\sigma_{i}/\delta(b,i) } {\sigma_{j}/\delta(b,j)}\big)^2, ~ \forall i\ne j \ne b, \label{eq:budget_allocationOCBA1}\\
\tilde{l}_{b} &= \sigma_{b}\sqrt{\sum_{i = 1, i\ne b}^{k}  \frac{\tilde{l}_i^2}{(\sigma_{i})^2 } }, \label{eq:budget_allocationOCBA2}\\
\sum_{i = 1}\tilde{l}_i &= T \label{eq:budget_allocationOCBA3},
\end{align}	
where \(\sigma_{i}\) is the standard deviation of the \(i-\)th reward distribution.
As in Remark 2, $\sigma_i$ is assumed to be known, but in practice can be unknown and approximated by sample standard deviation $\hat{\sigma}_i = \sqrt{\frac{1}{l_i } \sum_{m = 1}^{l_i} (\hat{J}^m_i - \bar{J}_i)^2}$.

\begin{algorithm}
	\KwIn{Total sampling budget \(T\), initial sample size $n_0$}
	\KwOut{Index of optimal action \(\hat{b}\)}
	Sample each of the $k$ alternatives \(n_0\) times\;
	Set counter \(l_i \leftarrow n_0 ~ \forall i = 1,2,\dots,k\)\;
	\(l \leftarrow kn_0\)\;
	Calculate \(\bar{J}_i\) and \(\hat{\sigma}_{i}^2\), \(\forall i = 1,2,\dots,k\)\;
	\While{\(l<=T\)}{
		Compute new budget allocation
		\((\tilde{l}_1,\tilde{l}_2,\dots,\tilde{l}_{k})\) by solving eq. \eqref{eq:budget_allocationOCBA1}-\eqref{eq:budget_allocationOCBA3} with budget \(l+1\)\;
		Sample \(\hat{i} = \arg\max_{1\le i \le k} (\tilde{l}_i - l_i)\)\; 
		Update \(\bar{J}_{\hat{i}}\) (and \(\hat{\sigma}_{\hat{i}}^2\) if sample variance is used)\;
		\(l_{\hat{i}} \leftarrow  l_{\hat{i}} + 1\)\;
		\(l \leftarrow l+1\)\;
	}
	return \(\hat{b} = \arg \max_{1\le i \le k} \bar{J}_i\) \;
	\caption{One-stage OCBA}
	\label{alg:OCBA}
\end{algorithm}	
\begin{lemma}\label{lemma:OCBA_correct}
	Given a set of \(k\) normal random variables (actions) with mean \(J_i\) and variance \(\sigma_i^2\), \(i = 1,2,\dots, k\), where \(J_i.\) are also normally distributed. Suppose OCBA is run with sampling budget \(T\). Define the PCS
	\begin{align*}
	PCS = P \bigg[\bigcap_{i = 1,i\ne b}^k (\tilde{J}_b-\tilde{J}_i) \ge 0\bigg],
	\end{align*}
	where \(\tilde{J}_i\) is the posterior distribution of \(J_i\) given \(l_i\) samples \(\forall i = 1,2,\dots,k\).
	Then, \(PCS \rightarrow 1\) as \(T\rightarrow\infty\).
\end{lemma}
\begin{proof}
	The \(PCS\) can be lower bounded by APCS, i.e., by the Bonferroni inequality
	\begin{align*}
	PCS =& P \bigg[\bigcap_{i = 1,i\ne b}^k (\tilde{J}_b-\tilde{J}_i) \ge 0\bigg]\\
	\ge&1-\sum_{i=1,i\ne b}^{k}P \bigg[ \tilde{J}_b-\tilde{J}_i \le 0 \bigg]  \\
	=&APCS .
	\end{align*}
	Thus, to prove that \(PCS\rightarrow 1\), it suffices to prove \(APCS \rightarrow 1\), i.e.,
	\begin{align*}
	\sum_{i=1,i\ne b}^{k}P \bigg[ (\tilde{J}_b-\tilde{J}_i)\le 0   \bigg] \rightarrow 0 ~~ \text{as} ~~ T\rightarrow\infty .
	\end{align*}
	Based on the normality assumption, the posterior distribution is also normal, i.e., \(\tilde{J}_i \sim N(\bar{J}_i,\sigma_i^2/l_i)\). Thus, \(\tilde{J}_b-\tilde{J}_i \sim N(\bar{J}_b - \bar{J}_i,\sigma_b^2/l_b + \sigma_i^2/l_i)\). Therefore,
	\begin{align}\label{eq:Phi}
	\sum_{i=1,i\ne b}^{k}P \bigg[ (\tilde{J}_b-\tilde{J}_i)\le 0   \bigg] = & \sum_{i=1,i\ne b}^{k} \Phi(-\frac{ \bar{J}_b-\bar{J}_i }{\sqrt{\sigma_b^2/l_b + \sigma_i^2/l_i}}),
	\end{align}
	where \(\Phi\) is the cdf of the standard normal distribution.
	Since
	\begin{align*}
	\sum_{i=1}^{k}l_i =& T,\\
	\end{align*}
	then when \(T\rightarrow \infty,\) at least one of the actions will be sampled infinitely many times, i.e., there exists an index \(i\) such that \(l_i \rightarrow \infty\). Then there are two possible cases: \(i \ne b\) and \(i = b\). 
	
	Case 1: \(i \ne b\)\\
	According to eq. \eqref{eq:budget_allocationOCBA1},
	\begin{align*}
	l_j = \Bigg(  \frac{\sigma_{j}/\delta(b,j) } {\sigma_{i}/\delta(b,i)} l_i  \Bigg)^2,~ \forall j \ne i, j\ne b.
	\end{align*}
	Since \(\sigma_i\) and \(\delta(b,i)\) are bounded for all i, \(l_j\rightarrow\infty, ~\forall j \ne b\).\\
	Therefore, by eq. \eqref{eq:budget_allocationOCBA2}, \(l_b\rightarrow \infty\).
	Thus, \(l_i\rightarrow\infty\) for all \(i = 1,2,\dots, k\).
	
	Case 2: \(i = b\)
	According to \eqref{eq:budget_allocationOCBA2},
	\begin{align*}
	{l}_{b} &= \sigma_{b}\sqrt{\sum_{i = 1, i\ne b}^{k}  \frac{{l}_i^2}{(\sigma_{i})^2 } } \rightarrow \infty.
	\end{align*}
	Thus there exists an index \(i\ne b\) such that \(l_i \rightarrow \infty\). By a similar argument in Case 1, we can conclude that \(l_i\rightarrow\infty\) for all \(i = 1,2,\dots, k\).\\
	In either case, we have \(l_i\rightarrow\infty\) for all \(i = 1,2,\dots, k\). Additionally, since \(\bar{J}_b\) is defined to be the maximum of all \(\bar{J}_i\), i.e., \(\bar{J}_b-\bar{J}_i \ge 0\) for all \(i\ne b\), eq. \eqref{eq:Phi} becomes
	\begin{align*}
	\sum_{i=1,i\ne b}^{k}P \bigg[ (\tilde{J}_b-\tilde{J}_i)\le 0   \bigg] = & \sum_{i=1,i\ne b}^{k} \Phi(-\frac{ \bar{J}_b-\bar{J}_i }{\sqrt{\sigma_b^2/l_b + \sigma_i^2/l_i}}) \\
	&\rightarrow 0
	\end{align*}
	as desired.
\end{proof}
A corollary follows directly from the lemma.
\begin{corollary}\label{corollary:LLN}
	Suppose one-stage OCBA is run with budget \(T\). Then 
	\begin{align*}
	\bar{J}_i\rightarrow \mathbb{E}[J_i] ~ \text{w.p. \(1\) as } T\rightarrow \infty, ~ \forall i = 1,2,\dots, k.
	\end{align*}
\end{corollary}
The proof is a simple application of the strong law of large numbers, since \(l_b\rightarrow\infty\).
With this key lemma, we are ready to prove the first two theorems proposed in \Cref{sec:analysis}. We start with \Cref{thm:asym_unbias}.\\
\begin{proof}[Proof of \Cref{thm:asym_unbias}]
	The result can be proved by induction.\par 
	
	First observe that since \(N\rightarrow\infty\), each path is explored infinitely many times. Thus the number of samples in each stage also goes to infinity as \(N\rightarrow\infty\).\par 
	
	Suppose at some point of the algorithm, all nodes are expanded. If the current state node $\bx$ is at stage \(H-1\) (i.e., it will transit into a terminal node in the next transition), running Algorithm \ref{alg:OCBAselection} reduces to a single-stage problem, which is the same as OCBA in Algorithm \ref{alg:OCBA}. \(\hat{Q}(\bx, a)\) can be viewed as a set of alternatives for \(a\in A\). From \Cref{corollary:LLN}, it is straightforward that 
	\begin{align*}
	\lim_{N\rightarrow\infty} \bar{Q}(\bx, a) = Q(\bx, a).
	\end{align*}
	Therefore, since the reward function is bounded
	\begin{align*}
	\lim_{N\rightarrow\infty} \hat{V}(\bx) &= \lim_{N\rightarrow\infty}\max_{a\in A_{\bx}} \bar{Q}(\bx, a)  \\
											   &= \max_{a \in A_{\bx}} \lim_{N\rightarrow\infty} \bar{Q}(\bx, a)\\
											   &= V^*(\bx).
	\end{align*}
	Now suppose that the statement is true for all child state nodes $\by$ of a state $\bx$, i.e., \(\hat{V}(\by) \rightarrow V^*(\by) \) and $\by$ could be achieved from $\bx$. Then for $\bx$, the algorithm also reduces to OCBA. Thus from \Cref{corollary:LLN} again
	\begin{align*}
	\lim_{N\rightarrow\infty} \bar{Q}(\bx, a) 
	&= \lim_{N(\bx, a)\rightarrow\infty} \bar{Q}(\bx, a) \\
	&= \mathbb{E}[R(\bx, a)] + \mathbb{E}_{P(\bx, a)}[ V^*(\by) ] \\
	&= Q(\bx, a)
	\end{align*}
	for all child state-action pair $(\bx, a)$.
	It follows that 
	\begin{align*}
	\lim_{N\rightarrow\infty} \hat{V}(\bx) &\rightarrow V^*(\bx).
	\end{align*}
\end{proof}

\Cref{thm:asym_correct} is a direct result of \Cref{lemma:OCBA_correct}.
\begin{proof}[Proof of \Cref{thm:asym_correct}]
	Since we assume \(\hat{Q}(\bx, a)\) is normally distributed with known variance, the posterior distribution of \(\hat{Q}(\bx, a)\), i.e., \(\tilde{Q}(\bx, a)\), is also a normal random variable. Then, it follows directly from \Cref{lemma:OCBA_correct} that
	\begin{align*}
	P &\bigg[\bigcap_{a\in A_{\bx}, a\ne \hat{a}_{\bx}^*}^k (\lim_{N\rightarrow\infty}\tilde{Q}(\bx,\hat{a}_{\bx}^*)-\lim_{N\rightarrow\infty}\tilde{Q}(\bx, a)) \ge 0\bigg] = 1,\\
	& \forall i = 1,\dots,H, ~ \bx\in \bX, ~ a\in A.
	\end{align*}
\end{proof}

\section{Allocation strategy}
\begin{proof}[Proof of \Cref{thm:max_PCS}]
	The problem of maximizing APCS with budget constraint can be formulated as 
	\begin{align*}
	&\max_{\tilde{N}(\bx, a), a \in A_{\bx}} 1-\sum_{a\in A, a\ne \hat{a}_{\bx}^*} P\bigg[ \tilde{Q}(\bx,\hat{a}_{\bx}^*) \le \tilde{Q}(\bx, a) \bigg]\\
	& s.t. \sum_{a \in A_{\bx}}\tilde{N}(\bx, a) = N.
	\end{align*}
	With Lagrange multiplier \(\lambda\), the Lagrangian can be written as
	\begin{align*}
	L 
	&= 1 - \sum_{a\in A_{\bx}, a\ne \hat{a}_{\bx}^*}P\bigg[ \tilde{Q}(\bx,\hat{a}_{\bx}^*) \le \tilde{Q}(\bx, a) \bigg] + \lambda (\sum_{a \in A_{\bx}}\tilde{N}(\bx, a) - N) \\
	&= 1 - \sum_{a\in A_{\bx}, a\ne \hat{a}_{\bx}^*}\Phi(\frac{\bar{Q}(\bx, a) - \bar{Q}(\bx,\hat{a}_{\bx}^*)}{\sigma_{\bx}(a, \hat{a}_{\bx}^*) } ) + \lambda (\sum_{a \in A_{\bx}}\tilde{N}(\bx, a) - N) \\
	&= 1- \sum_{a\in A_{\bx}, a\ne \hat{a}_{\bx}^*} \Phi(-\frac{\delta_{\bx}(\hat{a}_{\bx}^*,a) }{\sigma_{\bx}(a, \hat{a}_{\bx}^*) } ) + \lambda (\sum_{a \in A_{\bx}}\tilde{N}(\bx,a) - N) ,
	\end{align*}
	where 
	\begin{align*}
	\sigma^2_{\bx}(a, \hat{a}_{\bx}^*) 
	&= \frac{\sigma^2(\bx, \hat{a}_{\bx}^*)}{N(\bx, \hat{a}_{\bx}^*)} + \frac{\sigma^2(\bx, a)}{N(\bx,a)}
	\end{align*}
	Apply Karush-Kuhn-Tucker (KKT) conditions \cite{boyd2004convex}:
	\begin{itemize}[leftmargin=*]
		\item primal feasible
		\begin{align}
		&N(\bx,a) \ge 0, \forall a\in A   \label{app:eq:KKT:primal_feasible1}\\
		&\sum_{a \in A_{\bx}}\tilde{N}(\bx,a) - N = 0 \label{app:eq:KKT:primal_feasible2},
		\end{align}
		
		\item stationarity
		\begin{align}
		&\frac{\partial L}{\partial N(\bx,a)}
		= \frac{\partial L}{\partial (-\frac{\delta_{\bx}(\hat{a}_{\bx}^*,a) }{\sigma_{\bx}(a, \hat{a}_{\bx}^*) })} \frac{\partial (-\frac{\delta_{\bx}(\hat{a}_{\bx}^*,a) }{\sigma_{\bx}(a, \hat{a}_{\bx}^*) })}{\partial \sigma_{\bx}(a, \hat{a}_{\bx}^*)  } \frac{\partial \sigma_{\bx}(a, \hat{a}_{\bx}^*) }{\partial N(\bx, a)} \nonumber \\
		&=0
		\end{align}
	\end{itemize}

Case 1: \(a \ne \hat{a}_{\bx}^*\) 
\begin{align}\label{app:eq:KKT:der_a}
&\frac{\partial L}{\partial N(\bx,a)} \nonumber \\
&=  \frac{\sigma^2(\bx, a) \delta_{\bx}(\hat{a}_{\bx}^*,a) }{ N^2(\bx,a)\sigma_{\bx}^3(a, \hat{a}_{\bx}^* )\sqrt{2\pi}} \exp{-\frac{\delta_{\bx}(\hat{a}_{\bx}^*,a)^2 }{\sigma_{\bx}^2(a, \hat{a}_{\bx}^*) } }+ \lambda \nonumber \\
&=0
\end{align}.

Case 2: \(a = \hat{a}_{\bx}^*\) 
\begin{align}\label{app:eq:KKT:der_a^*}
&\frac{\partial L}{\partial N(\bx,\hat{a}_{\bx}^*)}\nonumber \\
&=  \sum_{a\in A_{\bx}, a\ne \hat{a}_{\bx}^*}\frac{\sigma^2(\bx,\hat{a}_{\bx}^*) \delta_{\bx}(\hat{a}_{\bx}^*,a) }{ N^2(\bx,\hat{a}_{\bx}^*)\sigma_{\bx}^3(a, \hat{a}_{\bx}^* )\sqrt{2\pi}} \exp{-\frac{\delta_{\bx}(\hat{a}_{\bx}^*,a)^2 }{\sigma_{\bx}^2(a, \hat{a}_{\bx}^*) } }+ \lambda \nonumber\\
&= 0
\end{align}.

From \Cref{app:eq:KKT:der_a},
\begin{align}\label{app:eq:KKT:substitute}
 \frac{\delta_{\bx}(\hat{a}_{\bx}^*,a) }{\sigma_{\bx}^3(a, \hat{a}_{\bx}^* )\sqrt{2\pi}} \exp{-\frac{\delta_{\bx}(\hat{a}_{\bx}^*,a)^2 }{\sigma_{\bx}^2(a, \hat{a}_{\bx}^*) } } =  -\lambda\frac{N^2(\bx,a)}{\sigma^2(\bx,a)}.
\end{align}
Plug \Cref{app:eq:KKT:substitute} into \Cref{app:eq:KKT:der_a^*} yields
\begin{align*}
 \frac{\sigma^2(\bx, \hat{a}_{\bx}^*)}{N^2(\bx, \hat{a}_{\bx}^*)}\sum_{a\in A_{\bx}, a\ne \hat{a}_{\bx}^*} \lambda\frac{N^2(\bx, a)}{\sigma^2(\bx,a)} = \lambda,
\end{align*}
i.e.,
\begin{align}\label{app:eq:budget_allocation:2}
 N(\bx,  \hat{a}_{\bx}^*)
=
\sqrt{
\sigma^2(\bx,\hat{a}_{\bx}^*)\sum_{a\in A_{\bx}, a\ne \hat{a}_{\bx}^*} \frac{N^2(\bx, a)}{\sigma^2(\bx, a)} 
}.
\end{align}
After sufficiently large number of samples, we may conclude from \Cref{app:eq:budget_allocation:2} that our algorithm would focus more on sampling the sample optimal. Thus, we may assume that \(N(\bx, \hat{a}_{\bx}^*) \gg N(\bx,a)\) for all suboptimal actions \(a\in A_{\bx}\). \par 

Now, for two suboptimal actions \(a \ne \tilde{a} \ne \hat{a}_{\bx}^*  \), we have
\begin{align*}
&\frac{\sigma^2(\bx, a) \delta_{\bx}(\hat{a}_{\bx}^*,a) }{ N^2(\bx,a) (\frac{\sigma^2(\bx, \hat{a}_{\bx}^*)}{N(\bx, \hat{a}_{\bx}^*)} + \frac{\sigma^2(\bx, a)}{N(\bx,a)})^{3/2} } \exp{-\frac{\delta_{\bx}(\hat{a}_{\bx}^*,a)^2 }{ 2 (\frac{\sigma^2(\bx, \hat{a}_{\bx}^*)}{N(\bx, \hat{a}_{\bx}^*)} + \frac{\sigma^2(\bx, a)}{N(\bx,a)}) } }  \\
&=
\frac{\sigma^2(\bx,\tilde{a}) \delta_{\bx}(\hat{a}_{\bx}^*,\tilde{a}) }{ N\bx^2(\bx,\tilde{a}) (\frac{\sigma^2(\bx, \hat{a}_{\bx}^*)}{N\bx(\bx, \hat{a}_{\bx}^*)} + \frac{\sigma^2(\bx, \tilde{a})}{N(\bx,\tilde{a})})^{3/2} } \exp{-\frac{\delta_{\bx}(\hat{a}_{\bx}^*,\tilde{a})^2 }{ 2 (\frac{\sigma^2(\bx, \hat{a}_{\bx}^*)}{N(\bx, \hat{a}_{\bx}^*)} + \frac{\sigma^2(\bx, \tilde{a})}{N(\bx,\tilde{a})}) } }.
\end{align*}
Apply the \(N(\bx, \hat{a}_{\bx}^*) \gg N(\bx,a)\) assumption:
\begin{align*}
&\frac{\sigma^2(\bx,a) \delta_{\bx}(\hat{a}_{\bx}^*,a) }{ N^2(\bx,a) ( \frac{\sigma^2(\bx, a)}{N(\bx,a)})^{3/2} } \exp{-\frac{\delta_{\bx}(\hat{a}_{\bx}^*,a)^2 }{ 2 ( \frac{\sigma^2(\bx, a)}{N(\bx,a)}) } }  \\
&=
\frac{\sigma^2(\bx,\tilde{a}) \delta_{\bx}(\hat{a}_{\bx}^*,\tilde{a}) }{ N^2(\bx,\tilde{a}) (\frac{\sigma^2(\bx, \tilde{a})}{N(\bx,\tilde{a})})^{3/2} } \exp{-\frac{\delta_{\bx}(\hat{a}_{\bx}^*,\tilde{a})^2 }{ 2 (\frac{\sigma^2(\bx, \tilde{a})}{N(\bx,\tilde{a})}) } }.
\end{align*}
i.e.,
\begin{align*}
\Big(\frac{N(\bx,\tilde{a})}{N(\bx,a)} \Big)^{1/2}
&= \frac{\sigma^2(\bx,a)}{\sigma^2(\bx,\tilde{a})}  \frac{\delta_{\bx}(\hat{a}_{\bx}^*,\tilde{a}) }{\delta_{\bx}(\hat{a}_{\bx}^*,a) } \\
&\exp(\frac{\delta_{\bx}(\hat{a}_{\bx}^*,a)^2 }{ 2 ( \frac{\sigma^2(\bx, a)}{N(\bx,a)}) } - \frac{\delta_{\bx}(\hat{a}_{\bx}^*,\tilde{a})^2 }{ 2 (\frac{\sigma^2(\bx, \tilde{a})}{N(\bx,\tilde{a})}) }).
\end{align*}
Taking log on both sides yields
\begin{align*}
&\log{N(\bx,\tilde{a})}  - \log(N(\bx,a))
= 2\log{ \frac{\sigma^2(\bx,a)}{\sigma^2(\bx,\tilde{a})}  \frac{\delta_{\bx}(\hat{a}_{\bx}^*,\tilde{a}) }{\delta_{\bx}(\hat{a}_{\bx}^*,a) }} \\
&+\frac{\delta_{\bx}(\hat{a}_{\bx}^*,a)^2 }{  \frac{\sigma^2(\bx, a)}{N(\bx,a)} } - \frac{\delta_{\bx}(\hat{a}_{\bx}^*,a)^2 }{  \frac{\sigma^2(\bx, \tilde{a})}{N(\bx,\tilde{a})} }.
\end{align*}
When the number of samples is sufficiently large (\(N\rightarrow \infty \)), the log terms can be neglected compared to the terms linear in \(N(\bx,a)\) or \(N(\bx,\tilde{a})\). Therefore, removing the log terms yields,

\begin{align*}
\frac{\delta_{\bx}(\hat{a}_{\bx}^*,a)^2 }{  \frac{\sigma^2(\bx, a)}{N(\bx,a)} }
=  \frac{\delta_{\bx}(\hat{a}_{\bx}^*,a)^2 }{  \frac{\sigma^2(\bx, \tilde{a})}{N(\bx,\tilde{a})} },
\end{align*}
namely,
\begin{align*}
\frac{\tilde{N}(\bx,a)}{\tilde{N}(\bx,\tilde{a})} =& 
\Bigg(\frac{\sigma(\bx, a)/\delta_{\bx}(\hat{a}_{\bx}^*,a)  } {\sigma(\bx, \tilde{a})/\delta_{\bx}(\hat{a}_{\bx}^*,\tilde{a}) }\Bigg)^2, \nonumber \\
~ &\forall  a\, , \tilde{a}\ne \hat{a}_{\bx}^*.
\end{align*}
\end{proof}

\section{Performance bound analysis}\label{app:sec:performance_analysis}
\begin{proof}[Proof of \Cref{thm:PCS_bound}]
	When the number of samples at node $\bx$ is large, we assume that \(N(\bx,a)\) satisfies \Crefrange{eq:budget_allocation1}{eq:budget_allocation2}.\par 
	
	From \Cref{eq:budget_allocation1}, we have 
	\begin{align}\label{eq:app:N_substitute}
	N(\bx, \tilde{a}) &= \Big( \frac{\sigma(\bx, \tilde{a}) \delta_{\bx}(\hat{a}_{\bx}^*,a)  }{\sigma(\bx, a) \delta_{\bx}(\hat{a}_{\bx}^*,\tilde{a})  }\Big)^2 N(\bx, a), \\
	&\forall \tilde{a}, a \ne \hat{a}^*_{\bx}. \nonumber
	\end{align}
	In this way, we can express the budget allocation to any suboptimal action \(\tilde{a}\) as the product of the budget allocation to a particular suboptimal action \(a\) and the factor
	\begin{align*}
	r_{\bx}( \tilde{a}, a) = \Big( \frac{\sigma(\bx, \tilde{a}) \delta_{\bx}(\hat{a}_{\bx}^*,a)  }{\sigma(\bx, a) \delta_{\bx}(\hat{a}_{\bx}^*, \tilde{a})  }\Big)^2 .
	\end{align*}
	From \Cref{eq:budget_allocation2}:
	\begin{align*}
	N(\bx, \hat{a}^*_{\bx}) = \sigma(\bx, \hat{a}^*_{\bx}) \sqrt{\sum_{\tilde{a} \in A_{\bx}, \tilde{a} \ne \hat{a}^*_{\bx}}  \frac{(N(\bx, \tilde{a})^2  )}{\sigma^2(\bx, \tilde{a})} }.
	\end{align*}
	Substitute \(N(\bx, \tilde{a})\) from \Cref{eq:app:N_substitute} yields
	\begin{align*}
	N(\bx, \hat{a}^*_{\bx}) = N(\bx, a) \sigma(\bx, \hat{a}^*_{\bx}) \sqrt{\sum_{\tilde{a} \in A_{\bx}, \tilde{a} \ne \hat{a}^*_{\bx}}  \frac{(r_{\bx}(\tilde{a}, a ) )^2 }{\sigma^2(\bx, \tilde{a})} } ,
	\end{align*}
	i.e.,
	\begin{align*}
	N(\bx, a) =  \frac{N(\bx, \hat{a}^*_{\bx})}{\sigma(\bx, \hat{a}^*_{\bx}) \sqrt{\sum_{\tilde{a} \in A_{\bx}, \tilde{a} \ne \hat{a}^*_{\bx}}  \frac{(r_{\bx}(\tilde{a}, a ) )^2 }{\sigma^2(\bx, \tilde{a})} } }.
	\end{align*}
	Since PCS is lower bounded by APCS, and the posterior \(\tilde{Q}(\bx,a)\) is normally distributed with 
	\begin{align*}
	\tilde{Q}(\bx,a) \sim N(\bar{Q}(\bx,a), \frac{\sigma^2(\bx, a)}{N(\bx,a)}),
	\end{align*}
	then
	\begin{align*}
	PCS 
	&\ge APCS \\
	&= 1-\sum_{a\in A_{\bx}, a\ne \hat{a}_{\bx}^*} P\bigg[ \tilde{Q}(\bx,\hat{a}_{\bx}^*) \le \tilde{Q}(\bx,a) \bigg] \\
	&= 1- \sum_{a\in A_{\bx}, a\ne \hat{a}_{\bx}^*} \Phi(\frac{\bar{Q}(\bx,a) - \bar{Q}(\bx,\hat{a}_{\bx}^*)}{\sigma_{\bx}(a, \hat{a}_{\bx}^*) } )\\
	&= 1- \sum_{a\in A_{\bx}, a\ne \hat{a}_{\bx}^*} \Phi(-\frac{\delta_{\bx}(\hat{a}_{\bx}^*,a) }{\sigma_{\bx}(a, \hat{a}_{\bx}^*) } ),
	\end{align*}
	where the second equality is because \(\tilde{Q}(\bx,\hat{a}_{\bx}^*) - \tilde{Q}(\bx,a)\) is normally distributed with mean \(\bar{Q}(\bx,a) - \bar{Q}(\bx,\hat{a}_{\bx}^*)\) and variance 
	\begin{align*}
	\sigma^2_{\bx}(a, \hat{a}_{\bx}^*) 
	&= \frac{\sigma^2(\bx, \hat{a}_{\bx}^*)}{N(\bx, \hat{a}_{\bx}^*)} + \frac{\sigma^2(\bx, a)}{N(\bx,a)} \\
	&= \frac{1}{N(\bx, \hat{a}_{\bx}^*)} \Bigg( \sigma^2(\bx, \hat{a}_{\bx}^*) + \\
	 &\sigma(\bx, \hat{a}_{\bx}^*) \sigma^2(\bx, a) \sqrt{\sum_{\tilde{a} \in A_{\bx}, \tilde{a} \ne \hat{a}^*_{\bx}}  \frac{(r_{\bx}(\tilde{a}, a ) )^2 }{\sigma^2(\bx, \tilde{a})} }   \Bigg).
	\end{align*}
	Apply inequality \(\sqrt{\sum_{i = 1}^{n}c_i^2} \le \sum_{i=1}^{n}\sqrt{c_i^2} = \sum_{i=1}^{n} c_i  \) for positive numbers $c_i$'s yields
	\begin{align*}
	\sigma^2_{\bx}(a, \hat{a}_{\bx}^*) 
	&\le \frac{1}{N(\bx, \hat{a}_{\bx}^*)} \Bigg( \sigma^2(\bx, \hat{a}_{\bx}^*) + \\
	&\sigma(\bx, \hat{a}_{\bx}^*) \sigma^2(\bx, a) \sum_{\tilde{a} \in A_{\bx}, \tilde{a} \ne \hat{a}^*_{\bx}}  \frac{ r_{\bx}(\tilde{a}, a )  }{\sigma(\bx, \tilde{a})}   \Bigg).
	\end{align*}
	Since APCS is decreasing in \(\sigma^2_{\bx}(a, \hat{a}_{\bx}^*) \), we have
	\begin{align*}
	&PCS 
	\ge 1- \\
	&\sum_{a\in A_{\bx}, a\ne \hat{a}_{\bx}^*} \Phi \bigg(-\frac{\delta_{\bx}(\hat{a}_{\bx}^*,a)  \sqrt{N(\bx, \hat{a}_{\bx}^*) }}{ \sqrt{ \sigma^2(\bx, \hat{a}_{\bx}^*) + 
			\sigma(\bx, \hat{a}_{\bx}^*) \sigma^2(\bx, a) \sum_{\tilde{a} \in A_{\bx}, \tilde{a} \ne \hat{a}^*_{\bx}}  \frac{ r_{\bx}(\tilde{a}, a )  }{\sigma(\bx, \tilde{a})}   } } \bigg)
	\end{align*}
	as desired
\end{proof}
\end{appendices}



%

%

\begin{IEEEbiographynophoto}{Yunchuan Li}
received a bachelor's degree in Automation from University of Electronic Science and Technology of China (UESTC). He is currently a Ph.D. candidate in the Department of Electrical and Computer Engineering at the University of Maryland, College Park. His research focuses on optimization and control with applications to operations research problems.
\end{IEEEbiographynophoto}


\begin{IEEEbiographynophoto}{Michael C. Fu}
(S’89–M’89–SM’06–F’08) received degrees in mathematics and EECS from MIT in 1985
and a Ph.D. in applied math from Harvard in 1989. Since 1989, he has been at
the University of Maryland, College Park, currently holding the Smith Chair
of Management Science. He also served as the Operations Research Program
Director at the National Science Foundation. His research interests include
simulation optimization and stochastic gradient estimation. He is a Fellow of
the Institute for Operations Research and the Management Sciences (INFORMS).
\end{IEEEbiographynophoto}

\begin{IEEEbiographynophoto}{Jie Xu}
(S’01-M’09-SM’17)
	received the Ph.D. degree in industrial engineering and management sciences from Northwestern University, the M.S. degree in computer science from The State University of New York, Buffalo, the M.E. degree in electrical engineering from Shanghai Jiaotong University, and the B.S. degree in electrical engineering from Nanjing University. He is currently an Associate Professor of Systems Engineering and Operations Research at George Mason University. His research interests are data analytics, stochastic simulation and optimization, with applications in cloud computing, manufacturing, and power systems. 
\end{IEEEbiographynophoto}

\end{document}